%% file: Final.tex
\documentclass[runningheads]{llncs}
\usepackage[T1]{fontenc}
\usepackage[english]{babel}
\usepackage{graphicx}
\usepackage{amssymb}
\usepackage{amsmath}
\usepackage[noxy]{virginialake}
\vlnostructuressyntax
\usepackage{tikz}
\usepackage{cmll}
\usepackage{adjustbox}
\usepackage{xspace}
\usepackage{cleveref}
\usepackage{proofzilla}
\usepackage{thm-restate}

\input{macros}
\def\lcalc{$\lambda$-calculus\xspace}

\def\mlcalc{modal $\lambda$-calculus\xspace}

\def\letcon{\mathsf{Let}}

\def\defn#1{\emph{#1}}

\usepackage{esvect}
\def\myvec#1{\vv#1}
\newvertex{vv}{v}{}

\def\myparagraph#1{\noindent\textbf{#1}.}

\begin{document}

%
\title{
	Canonicity of Proofs in Constructive Modal Logic
	\thanks{
		The first author  is supported by 
		Villum Fonden, grant no. 50079.
		The second author is supported by
		the PRIN project
		RIPER (No. 20203FFYLK)
		The third author  is supported by 
		the US Air Force Office for Scientific Research under award number FA9550-21-1-0007.
	}
}
%
%
\author{
	Matteo Acclavio		\inst{1}\and
	Davide Catta		\inst{2}\and
	Federico Olimpieri	\inst{3}	
}
\authorrunning{M. Acclavio et al.}
%
\institute{
	University of Southern Denmark, Odense, Denmark 
\and
	Università degli studi di Napoli, Federico II, Naples, Italy
\and
	University of Leeds, Leeds, UK
}
\maketitle              
\begin{abstract}
	In this paper we investigate the Curry-Howard correspondence for constructive modal logic in light of the 
	gap between the proof equivalences enforced 
	by the lambda calculi from the literature and 
	by the recently defined winning strategies 
	for this logic.
	
	We define a new lambda-calculus for a minimal constructive modal logic 
	by enriching the calculus from the literature with additional reduction rules
	and we prove normalization and confluence for our calculus.
	We then provide a typing system in the style of focused proof systems allowing us to provide a unique proof for each term in normal form, and we use this result to show a one-to-one correspondence between terms in normal form and winning innocent strategies.
\end{abstract}

\section{Introduction}

Proof theory is the branch of mathematical logic whose aim is studying the properties of logical arguments (i.e., proofs) as well as the structure of proofs and their invariants.
For this purpose, 
the most used representations of proofs are based on tree-like data structures inductively defined using inference rules of a proof system.%
\footnote{
	It is worth noting that some proof systems 
	(in the sense of \cite{cook:reckhow:79}) 
	allows to represent proofs using structures such as infinite trees (for non-well-founded proof systems, see, e.g., \cite{das:pous:non-wellFULL}), graphs (see proof nets \cite{girard:87,girard:96:PN}, combinatorial proofs \cite{hughes:pws,hughes:pws} or  proof diagrams \cite{acc:PD}) or structures defined in a compositional way (see open deduction \cite{gug:gun:par:2010} and deep inference \cite{tub:str:esslli19})
}
 \emph{Natural deduction} and \emph{sequent calculus} are among the most used proof systems due to their intuitive representation.
Both these proof systems were originally devised by Gentzen in order to prove the consistency of first-order arithmetic. 
Their versatility resulted in their employment for a wide variety of logics.

However, having formalisms able to represent proofs is not enough to define ``what is a proof'' since different derivations, or derivations in different proof systems, could represent the same abstract object.
A notion of \emph{proof identity} is therefore required to define a proof as a proper mathematical entity \cite{dosen:BSL}.
Such a notion of identity is provided by delineating the conditions under which two distinct formal representations of a proof represent the same logical argument. 
The definition of these conditions are often driven by semantic considerations (by performing specific transformations on two derivations, they can be transformed to the same object) or intuitive ones (two derivations only differ for the order in which the same rules are applied to the same formulas).

Natural deduction is often considered a satisfactory formalism since it allows to define a more canonical representation of proofs with respect to sequent calculus:  sequent  calculus derivations differing because of some rules permutations are represented (\textit{via} a standard translation) by the same  natural deduction derivation.
Moreover, natural deduction provides a one-to-one correspondence 
between derivations and lambda-terms, 
called the \emph{Curry-Howard correspondence} \cite{sorensen2006lectures}.
%

\myparagraph{Constructive Modal Logic}
Classical modal logics are obtained by extending \emph{classical logic} with unary operators, called \emph{modalities}, that qualify the truth of a judgment.
The most used modalities are the $\lbox$ (called \emph{box}) and its dual operator $\ldia$ (called \emph{diamond}) which are usually interpreted as \emph{necessity} and  \emph{possibility}.
According to the interpretation of such modalities, modal logics find applications, for example, in knowledge representation~\cite{knowledgerepresentation}, artificial intelligence~\cite{MEYERAI} and the formal verification of computer programs~\cite{Pnueli77,EmersonC82,Kozen83}.
The work of Fitch~\cite{fitch:48} initiated the investigation of the proof theory of modal logics extending intuitionistic logic, leading to numerous results on the topic \cite{prawitz:65,heilala2007bidirectional,mend:multimodal,fairtlough1997propositional,kojima2012semantical}.

In particular, the Curry-Howard correspondence has been extended to various constructive modal logics \cite{AlechinaMendlerDepaiva,bel:deP:rit:extended,kakutanilambda,Kavvos20,pfenningS4,pfenningStaged}. 
Intuitionistic logic can be extended with modalities in different ways (for an overview see  \cite{simpson:phd}):
while in classical logic axioms involving only $\lbox$ provide also description of the behavior of $\ldia$, for intuitionistic logic this is no more the case since the duality of the two modalities does not hold anymore.
This leads to different approaches.
\emph{Constructive modal logics} consider minimal sets of axioms to guarantee the definition of the behaviors of the $\lbox$ and $\ldia$ modalities.
A second approach, referred to as \emph{intuitionistic modal logic}, considers additional axioms in order to validate the G\"{o}del-Gentzen translation  \cite{das:mar:blog}.
In this work we consider a minimal fragment of the constructive modal logic $\CK$ 
only containing the implication $\imp$ and the modality $\lbox$.
This fragment is enough to define types for a \lcalc with a $\letcon$ constructor \cite{AlechinaMendlerDepaiva}
which can be interpreted as an explicit substitution
and, for this reason,
we more concisely denote by 
$\lambdabox{N}{M_1,\ldots M_n}{x_1,\ldots,x_n}$ 
instead of 
$\lambdalet {M_1,\ldots M_n} {x_1,\ldots,x_n} N $.

Recent works on the the proof equivalence of constructive modal logics \cite{acc:str:AIML22} expose a complexity gap between the proof equivalences induced by the natural deduction (\cite{bel:deP:rit:extended}) and winning innocent strategies (\cite{acc:cat:str:games}) for this logic.
This discrepancy cannot be observed in intuitionistic propositional logic where there are one-to-one correspondences between natural deduction derivations, lambda terms and innocent winning strategies.
In particular, in the logic $\CK$ we observe
sequent calculus proofs 
which correspond to the same winning strategy 
but which cannot be represented by the same natural deduction derivation in the systems provided in 
\cite{bel:deP:rit:extended,kakutanilambda}
(or equivalently corresponding to different \mlts).
By means of example, 
consider the  terms  
$\lambdabox{x}{z}{x}$
and
$\lambdabox{x}{z,w}{x,y}$
and their (unique) typing derivations shown in \Cref{fig:introDN}
(see \Cref{fig:typing} for the typing system).
\begin{figure}
	\centering
	\adjustbox{max width=.9\textwidth}{$\begin{array}{c@{\qquad}c@{\qquad}c}
			\vlderivation{
				\vliin{\dnbox}{}{
					z: \lbox a, w: \lbox b\vdash \lambdabox{x}{z}{x} : \lbox a
				}{ 
					\vlin{\Idrule}{}{z: \lbox a, w: \lbox b\vdash z:\lbox a}{\vlhy{}}
				}{
					\vlin{\Idrule}{}{x: a, y: b \vdash x : a}{\vlhy{}}
				}
			}
			\\
			\vlderivation{
				\vliiin{\dnbox}{}{
					z: \lbox a, w: \lbox b \vdash \lambdabox{x}{z,w}{x,y} : \lbox a
				}{ 
					\vlin{\Idrule}{}{z: \lbox a, w: \lbox b \vdash z:\lbox a}{\vlhy{}}
				}{
					\vlin{\Idrule}{}{z: \lbox a, w: \lbox b \vdash w:\lbox b}{\vlhy{}}
				}{
					\vlin{\Idrule}{}{x: a, y: b \vdash x : a}{\vlhy{}}
				}
			}
		\end{array}$}
	\caption{The typing derivations of the \mlts $\lambdabox{x}{z}{x}$
		and
		$\lambdabox{x}{z,w}{x,y}$.}
	\label{fig:introDN}
\end{figure}
Intuitively, 
the two terms 
$\lambdabox{x}{z}{x}$
and
$\lambdabox{x}{z,w}{x,y}$
should be semantically equivalent 
since
the explicit substitution of the variable $y$ in the term $x$
is vacuous.
Said differently, if we explicit the substitution encoded by the constructor $\letcon$, 
both terms 
$\lambdabox{x}{z}{x}$ and $\lambdabox{x}{z,w}{x,y}$
should reduce to the term $z$.

In fact, this undesirable behavior disappear when considering the Winning Innocent Strategies for $\CK$ defined in \cite{acc:cat:str:games}.
In this syntax, both the above natural deduction derivations correspond 
to the same strategy below.
\begin{equation}\label{ex:WIS}
	\strat=\Set{\emptyseq, \iseven a , \iseven a\isodd a}
	\quad
	\mbox{
	over the arena
	}
	\quad
	\arof{\lbox a, \lbox b \vdash \lbox a}
	\quad=\quad
	\begin{array}{c@{\qquad}c@{\qquad}c}
		\vlbox3	&	\vlbox1			&	\vlbox0
		\\[10pt]
		\vb1	&	\isodd{\va1}	&	\iseven{\va0}
		\\
	\end{array}
	\multiAedges{lbox1,a1}{lbox0,a0}
	\bentAedges{lbox3/lbox0/20,b1/a0/20,lbox3/a0/20,b1/lbox0/0}
	\Medges{lbox1/a1,lbox3/b1,lbox0/a0}
\end{equation}

\myparagraph{Contribution}
In this paper we define a new \mlcalc for $\CK$ by considering additional rewriting rules that allow us to retrieve a one-to-one correspondence between terms in normal form and winning innocent strategies, that is, providing more canonical representatives for proofs with respect to natural deduction and \mlts defined in the literature.
From the technical point-of-view, 
we obtain this result by
extending 
the operational semantics of the \mlcalc with the appropriate new reduction rules for the explicit substitution encoded by the $\letcon$, 
dealing with contraction and weakening operating on the variables bound by the $\letcon$.
We call this set of rules the $\kappa$-reduction, which we show to be strongly normalizing using elementary combinatorial methods. 
In order to deal with the interaction of the $\eta$-reduction with $\beta$-reduction, we define a restricted $\eta$-reduction following an approach similar to the one used in \cite{mints:closed,kes:fix,jay:eta}.
We prove strong normalization and confluence for our new operational semantics.

After proving confluence and strong normalization for our \mlcalc, 
we provide a canonical typing system inspired by focused sequent calculi (see, e.g., \cite{andreoli:01}) providing a unique typing derivation for each term in normal form.
We conclude by establishing a one-to-one correspondence between the winning strategies defined in \cite{acc:cat:str:games}
and 
proofs of this calculi,
therefore with terms in normal form.

\myparagraph{Related Work} 
To the best of our knowledge, 
the first paper proposing a Curry-Howard correspondence for the logic $\CK$ is \cite{bel:deP:rit:extended}.
In this work, the authors provide a natural deduction system for the logic $\CK$ by enriching the standard system for intuitionistic propositional logic with a generalized elimination rule capable of taking into account the behavior of the $\lbox$-modality.
At the level of lambda calculus, they enrich the syntax of terms by adding a new constructor $\letcon$ defined as follows:
\def\lambdafoot#1#2#3{\letcon \;#1\; \mathsf{be}\; #2 \;\mathsf{in}\;#3\;}
\begin{equation}\label{eq:let}
	\adjustbox{max width=.9\textwidth}{$
		\lambdafoot{x_1,\ldots x_n}{N_1, \ldots, N_n}{M}
		\quad
		\left(
		\mbox{which we denote $M \boxpar{N_1,\ldots, N_n/x_1,\ldots, x_n}$}
		\right)
	$}
\end{equation}
providing a notation which can be interpreted as an explicit substitution of the variable $x_i$ with the term $N_i$ for all occurrences of $x_1 \ldots, x_n$ inside a term $M$.
For this calculus, the authors only consider the usual $\eta$ and $\beta$ reductions plus the following reduction:
\begin{equation*}
	\adjustbox{max width=.95\textwidth}{$
		\begin{array}{c}
			\lambdafoot{y}{P}{(\lambdafoot{x}{N}{M})} 
			\rightsquigarrow
			\lambdafoot{x}{(\lambdafoot{y}{P}{N})}{(\lambdafoot{x}{x}{M})}
			\\
			\left(
			\mbox{in our syntax this reduction is written as}
			\quad
			\lambdabox{\lambdabox{M}{N}{x}}{P}{y} 
			\rightsquigarrow
			\lambdabox{\lambdabox{M}{x}{x}}{\lambdabox{N}{P}{y}}{x}
			\right)
		\end{array}
	$}
\end{equation*}

In \cite{kakutanilambda} the author considers the usual $\eta$ and $\beta$ reduction with 
an
the following 
additional $\beta$-reduction rule specifically designed to handle the explicit substitution construct. 
\begin{equation}
	\lambdabox{M}{
		\myvec{P},\rclr{ \lambdabox{R}{\myvec{N}}{\myvec{z}}}, \myvec{Q}
	}{
		\myvec{x},\rclr{y},\myvec{w}
	} 
	\bred[2] 
	\lambdabox{
		M\subst{R}{y}
	}{
		\myvec{P},\myvec{N},\myvec{Q}}{\myvec{x},\myvec{z},\myvec{w}
	}
\end{equation}
In the same paper, the author provides a detailed proof of 
strong normalization and confluence for modal lambda terms with respect to the standard $\eta$ and $\beta$ reduction, plus this new $\beta_2$ reduction.
However, also this calculus does not manage to fix the aforementioned problem with canonicity.

An alternative natural deduction system (and $\lambda$-calculus) is proposed in \cite{Kavvos20}, where the symmetry between elimination and introduction rules typical of natural deduction is restored.
However, this result requires to define a sequent calculus 
where
sequents have a more complex structure (dual-contexts),
and lacks an in-depth study of the operational semantics because the $\eta$-expansion is not considered in the calculus.

\textbf{Outline of the paper.}
In \Cref{sec:intro} we recall the definition of the fragment of the logic $\CK$ we consider in this paper, as well as the main results on the proof theory for this logic, its natural deduction and lambda calculus.
In \Cref{sec:new} we define the \mlcalc we consider in this paper, proving its strong normalization and confluence properties.
In \Cref{sec:foc} we provide a typing system in the style of focused sequent calculi, where we are able to narrow the proof search of the type assignment of our normal terms to a single derivation.
In \Cref{sec:games} we recall the definition of the game semantics for the logic we consider and we prove the one-to-one correspondence between terms in normal form and winning strategies.

\section{Preliminaries}\label{sec:intro}

In this section we recall the definition of the (fragment of the) constructive modal logic $\CK$ we consider in this paper, and we recall the definition and some terminology for \mlts.
We are interested in a minimal constructive modal logic whose \defn{formulas} are defined from a countable set of propositional variables $\mathcal{A}=\set{a,b,c,\ldots}$ using the following grammar:
\begin{equation}\label{eq:form}
	A \coloneqq  a \mid   ( A\imp A)       \mid \lbox A\ 
\end{equation}

We say that a formula is \defn{modality-free} if it contains no occurrences of the modality $\lbox$. 
A formula is a \defn{$\imp$-formula} if it is  of the form $A\imp B$. 
In the following  
we use  Krivine's convention~\cite{KrivineBook} and write $(A_1, \ldots, A_n )\to C$ as a shortcut for $(A_1\to( \cdots \to (A_n \to C)\cdots ))$
A \defn{sequent} is an expression $\Gamma\vdash C$ where $\Gamma$ is a finite (possibly empty) list of formulas and $C$ is a formula. 
If $\Gamma=A_1,\ldots ,A_n$ and $\sigma$ a permutation over $\intset1n$, then
we may write $\sigma(\Gamma)$ to denote $A_{\sigma(1)},\ldots ,A_{\sigma(n)}$.

In this paper we consider the logic $\CK$ 
defined by extending 
the conjunction-free and disjunction-free fragment of
intuitionistic propositional logic
with the modality $\lbox$ whose behavior is defined by the \defn{necessitation rule} and the axiom $\mathsf K_1$ below.
$$
\mathsf{Nec}\coloneqq \mbox{if $A$ is provable, then also $\lbox A$ is}
\qquad\qquad
\mathsf K_1\coloneqq \lbox(A\imp B) \imp (\lbox A \imp \lbox B)
$$
The sequent calculus $\SCK$, whose rules are provided in \Cref{fig:seqCK}, is a sound and complete proof system for the logic $\CK$. This system have been extracted from the one presented in \cite{kuz:mar:str:Justification} and satisfies cut-elimination.

\begin{figure}[t]
	\centering
	\def\myskip{\hskip1em}
	\adjustbox{max width=\textwidth}{$\begin{array}{c@{\myskip}c@{\myskip}c@{\myskip}c@{\myskip}c@{\myskip}c@{\myskip}c@{\myskip}|@{\myskip}c}
			\vlinf{\AXrule}{}{a\vdash a}{}
			&
			\vlinf{\exrule}{}{\sigma(\Gamma)\vdash C}{\Gamma\vdash C}
			&
			\vlinf{\rimprule}{}{\Gamma\vdash A\imp C}{\Gamma ,A\vdash C}
			&
			\vliinf{\limprule}{}{\Gamma, \Delta,A\imp B\vdash C}{\Gamma \vdash A}{B, \Delta \vdash C}
			\\\\
			\vlinf{\kbrule}{}{\lbox \Gamma \vdash \lbox  A}{\Gamma \vdash A}
			&
			\vlinf{\Wrule}{}{\Gamma, A\vdash C}{\Gamma\vdash C}
			&
			\vlinf{\Crule}{}{\Gamma, A\vdash C}{\Gamma, A, A\vdash C}
			&
			\vliinf{\cutr}{}{\Gamma , \Delta\vdash C}{\Gamma \vdash A}{\Delta, A\vdash C}
		\end{array}$}
	\caption{
		Sequent calculus rules of the sequent system $\SCK$,
		where $\sigma$ is a permutation over $\intset1n$
	}
	\label{fig:seqCK}
\end{figure}

\subsection{A Lambda Calculus for $\CK$}\label{subsec:DN}

\begin{figure}[t]
	\adjustbox{max width=\textwidth}{$
		\begin{array}{c}
			\vlinf{\Idrule}{i\in\intset1n}{x_1:A_1,\ldots, x_n:A_n \vdash x_i:A_i }{}
			\qquad
			\vlinf{\impintro}{}{\Gamma \vdash \tlambda{x}{}. M :  A\imp C}{\Gamma,x :  A \vdash M  :  C}
			\qquad
			\vliinf{\impelim}{}{\Gamma\vdash MN : C}{\Gamma\vdash N :   A}{\Gamma\vdash M : A\imp C}
			\\[15pt]
			\vliiiinf{\dnbox}{
				x_1,\ldots, x_n \text{ do not occur in $\Gamma$}
			}{ \Gamma \vdash \lambdabox{M}{N_1,\ldots ,N_n}{x_1,\ldots ,x_n} :  \lbox C
				
			}{ 
				\Gamma \vdash N_1  :  \lbox A_1
			}{\cdots}{
				\Gamma \vdash N_n  :  \lbox  A_n 
			}{
				x_1 :  A_1,\ldots, x_n  :  A_n \vdash M :  C
			}
		\end{array}
	$}
	\caption{
		Typing rules in the natural deduction system $\NDCK$ for \mlts.
	}
	\label{fig:typing}
\end{figure}

The set of (untyped) \defn{\mlts}
is defined inductively from 
a countable set of \defn{variables} $\variables=\set{x,y,\ldots}$ using the following grammar:%

\begin{equation*}
	\adjustbox{max width=\textwidth}{$
	M,N 			\coloneqq
	x 				\mid 
	\llambda{x}{M} 	\mid  
	(M N)	 		\mid 
	M \boxpar{\myvec N/\myvec x}
	\mbox{ where }
	\begin{cases}
		\myvec N=N_1,\ldots, N_n \mbox{ is a list of terms and}
		\\
		\myvec x = x_1,\ldots x_n \mbox{ is a list of distinct variables.}
	\end{cases}
	$}
\end{equation*}

\noindent modulo the standard $\alpha$-equivalence (denoted $\alphaeq$, see \cite{barendregtbookwithtypes})
and 
modulo the equivalence generated by the following permutations (for any $\sigma$ permutation over the set $\intset1n$) over the order of substitutions in the $\boxpar{\cdot/\cdot}$ constructor:
%
\begin{equation*}
	\adjustbox{max width=\textwidth}{$
		\begin{array}{c}
			\boxpar{\myvec N/\myvec x}
			\coloneqq
			\boxpar{N_1,\ldots, N_n/x_1,\cdots , x_n}
			=
			\boxpar{N_{\sigma(1)},\ldots, N_{\sigma(n)}/x_{\sigma(1)},\ldots , x_{\sigma(n)}}
			=:
			\boxpar{\sigma(\myvec N)/\sigma(\myvec x)}
		\\
			\mbox{for any $\sigma$ permutation over $\intset1n$.}
		\end{array}
	$}
\end{equation*}
%
As usual, application associates to the left, and has higher precedence than abstraction. 
For example, $\lambda xyz.xyz := \lambda x.(\lambda y.(\lambda z.((xy)z)))$. 
A \mlt is  a \defn{(explicit) substitution} if it is of the form $\lambdabox{M}{\myvec{N}}{\myvec{x}}$,
an \defn{application} if of the form $MN$,
and a \defn{$\lambda$-abstraction} if of the form $\lambda x.M$.

The set of \defn{subterms} of a term $M$ (denoted $ \mathsf{SUB}(M)$) 
is defined as follows:
\begin{equation*}
	\adjustbox{max width=\textwidth}{$
		\begin{array}{l}
			\subterm{x} = \set{  x }
			\quad ,\quad
			\subterm{\lambda x . M} = \subterm M \cup \set{ \lambda x . M }  
			\quad ,\quad
			\subterm{MN} = \subterm M \cup  \subterm N \cup \set{ MN }
			\;,
			\\
			\subterm{\lambdabox{M}{N_1, \dots, N_n}{x_1, \dots, x_n}} = 
			\subterm M 
			\cup 
			\left(\bigcup_{i\in\intset1n} \subterm{N_i}\right) 
			\cup
			\set{ \lambdabox{M}{N_1, \dots, N_n}{x_1, \dots, x_n}}
			\;.
		\end{array}
	$}
\end{equation*}

Its \defn{length} $\sizeof M$ and its set of \defn{free variables} $\FV M$ are defined as:
%
\begin{equation*}
	\adjustbox{max width=\textwidth}{$
		\sizeof M
		=
		\begin{cases}
			0								&\mbox{if } M=x					\\
			\sizeof N+1						&\mbox{if } M=\lambda x.N		\\
			\max\set{\sizeof N, \sizeof P}+1&\mbox{if } M=NP				\\
			\max\set{\sizeof N,\sizeof{P_1},\ldots,\sizeof{P_n}}+1&
			\mbox{if } M=\lambdabox{N}{\myvec{P}}{\myvec{x}}	\\
		\end{cases}
		\hskip2em
		\FV M
		=
		\begin{cases}
			\set x							&\mbox{if } M=x					\\
			\FV{N}\setminus\set{x}			&\mbox{if } M=\lambda x.N		\\
			\FV{N}\cup \FV{P}				&\mbox{if } M=NP				\\
			\bigcup_{i} \FV{P_i}			&\mbox{if } M=\lambdabox{N}{\myvec{P}}{\myvec{x}}							\\
		\end{cases}
	$}
\end{equation*}
We denote $\sizeinof M x$ the number of the occurrences of the free variable $x$ in a term $M$ and we may write $\sizeinof M x=0$ if $x\notin\FV M$
and we say that 
a term $M$ is \defn{linear} in the variables ${x_1, \ldots, x_n}$ if $\sizeinof M {x_i}=1$ for all $i\in\intset1n$.
We denote by $M\subst {N_1,\ldots, N_n}{x_1,\ldots, x_n}$ the result of the standard capture avoiding substitution of the occurrences of the variable $x_1,\ldots, x_n$ in $M$ with the term $N_1,\ldots, N_n$ respectively (see, e.g., \cite{terese}).

A \defn{variable declaration} is an expression 
$x :  A$
where $x$ is a variable and $A$ is a \defn{type}, that is, a formula as defined in \Cref{eq:form}. 
A \defn{(typing) context} is a finite list $\Gamma\coloneqq x_1: A_1,\ldots, x_n: A_n$ of distinct variable declarations.
Given a context 
$\Gamma= x_1 : A_1,\ldots, x_n : A_n$, 
we say that a variable $x$ \defn{appears} in $\Gamma$ if $x=x_i$ for a $i\in\intset1n$
and
we denote by $\Gamma, y : B$ the context ${ x_1 : A_1,\ldots, x_n : A_n,y: B}$ implicitly assuming that $y$ does not appear in $\Gamma$.
A \defn{\as} is an expression of the form $\Gamma \vdash M : A$ where $\Gamma$ is a context, $M$ a \mlt and $A$ a type.

\begin{definition}\label{def:derivation}
	Let $\Gamma\vdash M : A$ be an \as.
	A \defn{typing derivation} (or \defn{derivation} for short) 
	of $\Gamma\vdash M : A$ in $\NDCK$ 
	is a finite tree of \as 
	constructed using the rules in \Cref{fig:typing} in such a way 
	it has root $\Gamma \vdash M: A$
	and 
	each leaf is the conclusion of a $\Idrule$-rule.
	A \as is \defn{derivable} (in $\NDCK$) if there is a derivation with conclusion the given \as.

	We denote by $\mltset$ (resp. by $\smltset$ and $\amltset$) the set of \mlts (resp.~the set of substitutions and $\lambda $-abstractions in $\mltset$) admitting a derivable type assignment in $\NDCK$.
\end{definition}



\section{A New Modal Lambda Calculus}\label{sec:new}

In this section we define a new modal lambda calculus by enriching the operational semantics of the previous calculi with additional reduction rules aiming at recovering canonicity,
proving confluence and strong normalization properties. 


To define our term rewriting rules, 
we require special care when they are applied in a proper sub-term. This is due to the fact that the explicit substitution encoded by $\boxpar{\cdot/\cdot}$ could capture free variables.
For this reason, we introduce the notion of 
\defn{term with a hole} 
as a 
term of the form $\cC\conso$ containing a single occurrence of a special variable $\chole$.
More precisely, 
the set $\CwH$ of terms with a hole
and the two sets 
$\CwH[\eta_1]$ and $\CwH[\eta_2]$ 
of specific terms with a hole
are defined by the following grammars:
%
\begin{equation*}
	\adjustbox{max width=\textwidth}{$
		\begin{array}{c}
		\begin{array}{l@{\;:\;}c@{\;\coloneqq\;}r@{\;\mid\;}c@{\;\mid\;}ll@{\;\mid\;}l@{\mid\;}l@{\;}l}
			\CwH									&	
			\cC\conso 								&
			\chole		 							&
			\lambda x . \cC\conso 					&
			M\cC\conso 							\mid&
			\cC\conso M 							&
			\cC\conso [\myvec{M} / \myvec{x} ]			&
			\lambdabox{M}{ \myvec{N}_1 , \cC\conso, \myvec{N}_2}{\myvec{x}_1,  x, \myvec{x}_2}
			\\
			\CwH[\eta_1]							&	
			\cE\conso								&
			\chole		 							&
			\lambda x . \cE\conso					&
			M \cE\conso							\mid&
			\rclr{\cE'\conso} M  					&
			\lambdabox{\cE\conso}{\myvec{M}}{\myvec{x}} &
			\lambdabox{M}{ \myvec{N}_1 , \cE , \myvec{N}_2}{\myvec{x}_1,  x,  \myvec{x}_2}
			&
			\\
			\CwH[\eta_2]							&	
			\cD\conso								&	
			\chole		 							&
			\lambda x . \cD\conso					&
			M \cD\conso							\mid&
			\cD\conso M 							&
			\lambdabox{\cD\conso}{\myvec{M}}{\myvec{x}}	&
			\lambdabox{M}{ \myvec{N}_1 , \rclr{\cD'\conso} , \myvec{N}_2}{\myvec{x}_1,  x,  \myvec{x}_2} 								&
		\end{array}
		\\
		\mbox{with $\rclr{\cE'\conso}\neq \conso \neq \rclr{\cD'\conso}$}
		\end{array}
	$}
\end{equation*}
%
We denote by $\cC\cons M$ the term obtained by replacing the hole $\chole$ in $\cC\conso$ with the term $M$.
By means of example, 
if $\cC\conso=\chole$ then $\cC\cons M=M$ and if $\cE\conso= \lambdabox{(\lambda x.x N)}{\chole}{x}$ then $\cE\cons M=\lambdabox{(\lambda x.xN)}{M}{x}$.
The reduction relations of our calculus are provided in \Cref{fig:reductionContexts}, where the ground steps and the rules for extending them to specific contexts are provided.

\begin{figure}[t]
	\textbf{Ground Steps:}\\
	\adjustbox{max width=\textwidth}{$\begin{array}{r@{\;}r@{\;}ll}
			(\llambda x M)N 
			&
			\bred[1]
			& 
			M\subst Nx
			\\
			\lambdabox{M}{\myvec{P},\rclr{ \lambdabox{R}{\myvec{N}}{\myvec{z}}}, \myvec{Q}}{\myvec{x},\rclr{y},\myvec{w}} 
			&
			\bred[2]
			&
			\lambdabox{M\subst{R}{y}}{\myvec{P},\myvec{N},\myvec{Q}}{\myvec{x},\myvec{z},\myvec{w}}
			\\
			%
			M
			&
			\ered[1]
			&
			\llambda x M x  
			\qquad
			\mbox{if } \Gamma \vdash M :  A\imp B, \; x\notin \FV M  \mbox{ and } M \notin \amltset
			\\
			M
			&  
			\ered[2]
			&
			\lambdabox{x}{M}{x}
			\qquad
			\mbox{if } \Gamma \vdash	M : \lbox A, \; x\notin \FV M   \mbox{ and } M \notin \smltset
			\\
			\lambdabox{ M	}{\myvec P,\rclr{N}, \myvec Q	}{{\myvec{x}},\rclr{y},\myvec{z}}  
			&
			\kred[1]
			&
			\lambdabox{M}{\myvec{P},\myvec{Q}}{\myvec{x},\myvec{z}} 
			\qquad
			\mbox{if }\sizeinof{M}{y}=0
			\\
			\lambdabox{M}{\myvec P,\rclr{N,N},\myvec Q}{{\myvec x},\rclr{y_1,y_2},{\myvec{z}}}    
			&
			\kred[2]
			&
			\lambdabox{M\subst{v,v}{y_1,y_2}}{\myvec P,\rclr{N},\myvec Q 	}{	{\myvec x},\rclr{v},\myvec{z}}
			\qquad
			\mbox{with }v \mbox{ fresh}
		\end{array}$}	
	\\
	\\
	\textbf{Reduction Steps in Contexts:}
	\\
	\adjustbox{max width=\textwidth}{$
		\begin{array}{ccccc}
			\vlinf{}{i\in\set{1,2}}{ \cC \cons{M} \bred \cC \cons{N}  }{M \bred[i] N}   
			\quad
			\vlinf{}{i\in\set{1,2}}{ \cC \cons{M} \kred \cC \cons{N}  }{M \kred[i] N}
			&&
			\vlinf{}{}{ \cE \cons{M} \ered \cE \cons{N}  }{M \eredo N  }  
			&&
			\vlinf{}{}{ \cD \cons{M} \ered \cD \cons{N}  }{M \eredt N  }  
			\\ 
			\mbox{with} 
			\qquad
			\cC\conso\in\CwH
			\qquad\quad
			&\quand&
			\cE\conso\in\CwH[{\eta_1}]
			&\quand&
			\cD\conso\in\CwH[{\eta_2}]
		\end{array}
	$}
	
	\caption{
		Definition of the ground steps of the reduction relations,
		and the rules for their extension to terms with holes.
	}
	\label{fig:reductionContexts}
\end{figure}

\begin{remark}
	The term constructor $\letcon$ (i.e., $\boxpar{\cdot/\cdot}$ from \Cref{eq:let}) 
	plays no role in the standard $\eta$ and $\beta$ reduction rules from the literature, where it behaves as a black-box during reduction.
	The inertness of this constructor with respect to normalization is indeed what makes the lambda calculus in \cite{bel:deP:rit:extended,kakutanilambda} 
	unable to identify terms whose expected behavior is the same
	as, for example, the following pairs of terms:
	\begin{equation}\label{eq:WC}
		\begin{array}{c|c}
			\lambdabox{x}{v}{x}
			\quand
			\lambdabox{x}{v,w}{x,y}
			\qquad&\qquad
			\lambdabox{xyz}{v,v}{y,z}
			\quand
			\lambdabox{xyy}{v}{y}
		\end{array}
	\end{equation}

%
%
%
%
	Our operational semantics extends the one provided in \cite{kakutanilambda}.
	The novelty of our approach is 
	the definition of the $\kappa$-reduction 
	and 
	the restriction of the $\eta$-reduction. 
	The former is needed  to being able to identify \mlts with the same expected computational meaning, as the ones in \Cref{eq:WC}.
	The latter is carefully defined to avoid $\eta$-redexes that would make the reduction non-terminating, using a well-known technique in term rewriting theory (see, e.g., \cite{mints:closed,jay:eta}). 

	The need of these restrictions can be observed in the two following (unrestricted) $\eta$-reduction chains,
	which are both forbidden by our restricted rule from \Cref{fig:reductionContexts}.

\noindent\adjustbox{max width=\textwidth}{$\begin{array}{c}\\[-5pt]
	\begin{array}{cc}
		M \ered
		\llambda x M x
		\ered
		\llambda x (	\llambda y M y ) x  
		\ered 
		\dots  
		\\
		\mbox{whenever $\Gamma \vdash M : A \to B $}
	\end{array}
	\quand
	\begin{array}{c}
		M \ered
		\lambdabox{x}{M}{x}
		\ered
		\lambdabox{x}{\lambdabox{y}{M}{y}
		}{x} 
		\ered
		\dots 
		\\
		\mbox{whenever $\Gamma \vdash M : \lbox A$}
	\end{array}
\\[10pt]\end{array}$}

	Moreover, 
 	our definition rules out interactions between the $\eta$ and $\beta$ reductions
 	which could lead to infinite chains, as the ones shown below.
	$$
	\begin{array}{ccccc@{\alphaeq}c@{\qquad}l}
		\llambda x M 							&
		\ered 									&
		\llambda y (\llambda x M ) y 			&	 
		\bred									&
		\llambda y ( M \subst x y )				&
		\llambda x M 
		&
		\mbox{or}
		\\
		\lambdabox{x}{M}{x} 					&
		\ered 									&
		\lambdabox{x}{\lambdabox{y}{M}{y}}{x}	& 
		\bred 									&
		\lambdabox{y}{M  }{y} 					&
		\lambdabox{x}{M}{x} 
		&
		.
	\end{array}
	$$
\end{remark}

\begin{definition}
	We define the following reduction relations:
%
	\begin{equation}\label{eq:bekred}
		\bered=\bred\cup\ered
		\qquad
		\bkred=\ered\cup\kred
		\qquad
		\bekred=\bred\cup\ered\cup\kred
	\end{equation}
%
	For any $\xi \in\set{\beta,\eta, \kappa, \beta\eta,\beta\kappa, \beta\eta\kappa}$, 
	we denote 
	by $ \transc{\xi}$  its \defn{transitive closure}, 
	by $ \reflc{\xi} $  its \defn{reflexive closure},
	by $ \rtransc{\xi}$ its \defn{reflexive and transitive closure}, and 
	by $\equiv_\xi$ 	the \defn{equivalence relation} it enforces over terms, 
	that is, 
	its reflexive, symmetric and transitive closure.
	Given a term $ M $, 
	we denote by $ \nf{M}{\xi} $ 
	the set of its 
	\defn{$\rightsquigarrow_{\xi}$-normal form}. 
	A term $M$ is \defn{strongly normalizable} for $\rightsquigarrow_{\xi} $ if it admits no infinite $\rightsquigarrow_{\xi}$-chains 
	A reduction $\rightsquigarrow_{\xi}$ is \defn{strongly normalizing} if every term $ M$ is strongly normalizable for it. 
	A reduction $ \rightsquigarrow_{\xi} $ is \defn{confluent} if given $ M \rtransc{\xi} N_1  $ and $  M \rtransc{\xi} N_2 $ there exists a term $ N $ such that $ N_1 \rtransc{\xi} N $ and $ N_2 \rtransc{\xi} N$.
\end{definition}

The \emph{substitution} lemma and \emph{subject reduction} theorem holds for the reduction $\bekred$.
\begin{restatable}{lemma}{lemSubst}[Substitution Lemma]\label{lemma:substitution}
	Let 
	$\Gamma, x: B \vdash M: C$ and $\Gamma \vdash N : B $
	be derivable type assignments.
	Then $\Gamma ,x:B\vdash M \subst{N}{ x}: C$ is a derivable \as.
\end{restatable}
\begin{proof}
	The proof is by induction on $\sizeof{M}$. 
	\begin{itemize}
		\item 
		If $\sizeof M=1$ is a variable $z$ we either have that $M=x$ or $M=y$ for some $y\neq x$ that appears in $\Gamma$.
		Then 
		either $M\subst N x= N$ and $C=B$, 
		or $M\subst N x =M$. 
		In both cases $\Gamma \vdash M \subst{N}{ x}:C$ is derivable by hypothesis.

		\item

		If $\sizeof M \geq 1$ and 
		$M$ is an abstraction or an application, 
		then the proof is the same as in standard $\lambda$-calculus~\cite{terese}.
		If
		$M=\lambdabox{P}{T_1,\ldots,T_n}{x_1,\ldots x_n}$
		then $C=\lbox C'$
		and
		$M\subst N x =\lambdabox{P}{T_1 \subst {N}{x},\ldots T_n \subst{N}{x}}{x_1,\ldots x_n}$.
		Then, by definition of derivability, 
		we have
		$x\notin\set{x_1,\ldots, x_n}$ 
		and the
		type assignments
		$x_1 :  A_1,\ldots, x_n  :  A_n \vdash M :  C$
		and
		$\Gamma, x:B\vdash N_i  :  \lbox  A_i$
		are derivable
		for all $i\in\intset1n$ and for some $A_1,\ldots, A_n$.
		We can apply inductive hypothesis on
		$\Gamma, x:B\vdash N_i  :  \lbox  A_i$
		to deduce that the \as
		$\Gamma \vdash T_i \subst{N}{x} : \lbox A_i$ is derivable 
		for all $i\in\intset1n$.
		We conclude the existence the desired \as with bottom-mots rule a $\dnbox$
		with premises
		$x_1 :  A_1,\ldots, x_n  :  A_n \vdash \subst Mx :  C$
		and
		$\Gamma \vdash T_i \subst{N}{x} : \lbox A_i$ for all $i\in\intset1n$.
	\end{itemize}
\end{proof}

\begin{restatable}{theorem}{thmSubRed}\label{thm:subjr}
	Let $\Gamma \vdash M: C$ be derivable.
	If $M\bekred N$, then $\Gamma \vdash N: C$.  
\end{restatable}
\begin{proof}
	Because of \Cref{lemma:substitution}, it suffices to check the cases when $M$ reduces to $N$ in one ground step of $\bekred$:
	\begin{itemize}
		\item 
		if $M\bred[1] N$, then $ M = (\lambda x . P)Q$ and $ N = P \subst Q x $. 
		The case where $ M \bred[2] N $ uses a similar argument.
		The result follows the fact that
		if $\Gamma, x: B \vdash M: C$ and $\Gamma \vdash N : B $
		are derivable \as, then $\Gamma ,x:B\vdash M \subst{N}{ x}:C$ by \Cref{lemma:substitution}.
		
		\item 
		if $M\ered[1] N$, then $C=A\imp B$ and $N=\lambda x.Mx$.
		The result follows by applying the rule $\impintro$. 
		The case where $ M \ered[2] N $ uses a similar argument;
		
		\item 
		if $M\kred[1] N_1$, then
		$ M = \lambdabox{M'}{P_1,\ldots, P_k, \rclr{N}, P_{k+1},\ldots ,P_n}{x_1,\ldots, x_k, \rclr x , x_{k+1},\ldots, x_n} $ such that $ x$ is not free in $ M $,
		$C=\lbox B$,
		and
		$ N_1 = \lambdabox{M'}{\myvec{P}, \myvec{Q}}{\myvec{x},\myvec{y}} $.
		Then
		there are
		derivations for
		$\Gamma\vdash P_i:A_i$ for all $i\in\intset1n$ (for some $A_i$)
		and
		a derivation for
		$ x_1:A_1,\ldots,x_k: A_k, x : A, x_{k+1}:A_{k+1}\ldots , x_n:A_n \vdash M' : B$.
		Therefore
		we have a derivation for
		$ x_1: A_1,\ldots, x_n: A_n \vdash M' : B$
		since
		weakening is admissible (that is, 
		whenever $\Gamma,x:A \vdash M : C$ is derivable and $x$ does not occur free in $M$, then $\Gamma \vdash M : C$ is also derivable%
		\footnote{
			The admissibility of weakening is easily proven by induction on the size of a derivation.
		}%
		.
		Then we have a derivation of
		$\Gamma \vdash N: C$
		with bottom-most rule a $\dnbox$
		with right-most premise
		$ x_1 : A_1,\ldots, x_n: A_n\vdash M' : B$.
		and
		a premise
		$\Gamma\vdash P_i:A_i$ for each $i\in\intset1n$;

		\item 
		if $M\kred[2] N_1$, then
		we conclude similarly to the previous point
		since we have
		$$
		\hskip-1em
		M=\lambdabox{M'}{\myvec P,\rclr{N,N}, \myvec Q}{\myvec x,\rclr{y_1}, \rclr{y_2},\myvec{z}}
		\quand
		N_1=\lambdabox{ M \subst{y, y}{y_1, y_2}}{\myvec P,\rclr{N},\myvec Q }{	\myvec x,\rclr{y},\myvec{z}}
		\;.
		$$
		
	\end{itemize}
\end{proof}

We can prove local confluence of $\bekred$ by case analysis of the critical pairs using the following lemma.

\begin{restatable}{lemma}{lemRedsub}\label{lem:redsub} 
	Let $P,P'$ and $Q$ \mlts.
	If $ P \bekred P'$, then
	$P\subst {Q} x \bekred^\ast P'\subst {Q} x $.
	Moreover, 
	there is a $N_Q$ such that 
	$  Q \subst P x \bekred^\ast N_Q$
	and 
	$  Q\subst {P'} x \bekred^\ast N_Q$.
\end{restatable}
\begin{proof}
	
	To prove that $ P \subst Q x \bekred^\ast P' \subst {Q} x$ 
	we first prove that the result holds for the ground reduction steps:
	
	\begin{itemize}
		
		\item 
		If $ P = (\lambda y . M) N  \bred[1]  M \subst N y =P'$,
		then we have $ P \subst Q x = ((\lambda y . M) N )\subst Q x  =  (\lambda y . M \subst Q x) (N \subst Q x)$. 
		We conclude by associativity of substitution. 
		
		\item 
		If $ P \bred[1]  P'$, we conclude similarly to the previous case.

		\item 
		If $P \ered[1] \lambda y. P y=P'$, then 
		$(\lambda y. P y ) \subst {Q} x =\lambda y. (P \subst {Q} x) y  $.
		We conclude since,  definition of $\ered $, we have that $ P \subst Q x \subst q x  \ered \lambda y. (P \subst {Q} x) y  $. 
		
		\item 
		If $P \ered[2] \lambdabox y P y=P'$,
		then 
		$ P' \subst {Q} x =  \lambdabox y {(P \subst {Q} x)} y$.
		We conclude since, by definition of $\ered $, we have that $ P  \subst Q x \ered \lambdabox y {(P \subst {Q} x)} y $. 
		
		\item 
		If $ P = \lambdabox{ M'}{\myvec{P}, N, \myvec{Q}}{\myvec{x}, y, \myvec{z}} \kred[1] P' =  \lambdabox{ M'}{\myvec{P}, \myvec{Q}}{\myvec{x}, \myvec{z}}$,
		then 
		$  P \subst Q x = \lambdabox{ M' \subst Q x }{\myvec{P} \subst Q x , N \subst Q x, \myvec{Q} \subst Q x }{\myvec{x}, y, \myvec{z}}$. 
		We conclude since $ P' \subst Q x = \lambdabox{ M' \subst Q x }{\myvec{P} \subst Q x , \myvec{Q} \subst Q x }{\myvec{x}, \myvec{z}} $, thus $ P \subst Q x \kred[1] P' \subst Q x$.
				
		\item 
		If $ P =  \lambdabox{ M'}{\myvec{P}, N, N, \myvec{Q}}{\myvec{x}, y_1, y_2, \myvec{z}} \kred[2] P' =  \lambdabox{ M' \subst{v,v}{y_1, y_2}}{\myvec{P}, N, \myvec{Q}}{\myvec{x}, v, \myvec{z}}$, 
		then 
		$ P \subst Q x =  \lambdabox{ M' \subst Q x}{\myvec{P} \subst Q x, N \subst Q x, N \subst Q x, \myvec{Q} \subst Q x}{\myvec{x}, y_1, y_2, \myvec{z}}$ and $ P' \subst Q x =  \lambdabox{ {(M' \subst{v,v}{y_1, y_2})}\subst Q x}{\myvec{P} \subst Q x, N \subst Q x,  \myvec{Q} \subst Q x}{\myvec{x}, v, \myvec{z}}$.
		We conclude since 
		$ {(M' \subst{v,v}{y_1, y_2})}\subst Q x = {(M'\subst Q x )}\subst{v,v}{y_1, y_2} $, 
		thus  $ P \subst Q x \kred[2] P' \subst Q x$.
	\end{itemize}
	
	Then we conclude by showing that it also holds when reductions are applied in a context.
	
	\begin{itemize}
		\item 
		If $ P = \lambda y . M $ for a $M$ such that $ M \bekred M' $,
		then $P'=\lambda y.M'$  and we conclude by inductive hypothesis since $P\subst {Q} x=\lambda y. M\subst {Q} x \bekred \lambda y. M' \subst {Q} x$.
		
		\item 
		If $ P = MN  $, then $ P \subst Q x  = M \subst Q x N \subst Q x  $.
		In this case, 
		either $ M \bekred M'$ and $ P' = M'N $, 
		or $ N \bekred N' $  and $ P' = MN'$.
		We conclude 
		taking into account the restriction of the possible application of the reduction steps in a context.
		Note that without the restriction on $\ered[1]$ we could have had $ M \subst Q x \ered[1] \lambda y. (M \subst Q x) y $,
		and therefore $M \subst Q x N \ered[1] \lambda y. (M \subst Q x) y N \bred[1] M \subst Q x N$;
		
		\item 
		If $ P =  \lambdabox M N y $
		then $  P \subst Q x  =  \lambdabox {M \subst Q x} {N \subst Q x} y$. 
		In this case,  
		either $ M \bekred M' $ and $ P' =  \lambdabox {M'} N  y $,
		or $ N \bekred N' $ and $P' = \lambdabox {M} {N'}{y} $.
		We conclude 
		taking into account the restriction of the possible application of the reduction steps in a context.
		Note that without the restriction on  $\ered[2]$ we could have had  
		$ N \subst Q x \ered[2] \lambdabox y {N \subst Q x} y $,
		and therefore the following sequence of reductions.
		$$
		\lambdabox M {N \subst Q x} z \ered[2] \lambdabox M {\lambdabox y {N \subst Q x} y } z \bred[2] \lambdabox M {N \subst Q x} z 
		$$
		
	\end{itemize}
	
	The fact that,
	for each $Q$,
	there is a $N_Q$ such that 
	$  Q \subst P x \bekred^\ast N_Q$
	and 
	$  Q\subst {P'} x \bekred^\ast N_Q$
	is proven
	by induction on the structure of $Q$
	and 
	considering the restrictions on the definition of the rewriting steps in a context.

\end{proof}

\begin{proposition}\label{prop:locConf}
	The reduction $\bekred$ is locally confluent.
\end{proposition}
\begin{proof}
	We show that 
	if there are $M$, $N_1$ and $N_2$ with $N_1\neq N_2$
	such that 
	$ M \bekred N_1 $ and $ M \bekred N_2$, 
	then there exists $N $ such that $N_1 \rtransc{\beta \eta \kappa} N $ and $ \rtransc{\beta \eta \kappa} N$.
	Without loss of generality we have the following cases:
	
	\begin{enumerate}
		\item 
		if $M\bred[1]N_1$
		with
		$ M = (\lambda x . P) Q $ 
		and 
		$ N_1 =  P \subst  Q x$,
		then
		$N_2$ can only be obtained by applying $\bekred$ the subterms $P$ and $Q$ of $M$.
		We conclude by \Cref{lem:redsub};
		
		\item \label{localc:b2}
		if $M\bred[2]N_1$ with
		$ M = \lambdabox{M'}{\myvec{P}, \lambdabox{R}{\myvec{N}}{\myvec{z}}, \myvec{Q}}{\myvec{x}, y, \myvec{w}} $ 
		and with \\
		$ N_1 =  \lambdabox{M' \subst R y}{\myvec{P}, \myvec{N}, \myvec{Q}}{\myvec{x}, \myvec{z}, \myvec{w}}$,
		then $N_2$ must be a term obtained by applying $\bekred$ on $R$ or on one of the terms in 
		$ \myvec{ P}$, $\myvec{ N}$ or $\myvec{Q}$.
		We conclude again by \Cref{lem:redsub};

		\item
		if $M\ered[1] N_1$,
		then 
		$ \Gamma \vdash M  : A \imp B $
		and 
		$ N_1 = \lambda x . M x$.
		Therefore, for any $N_2$ such that 
		$M \bekred N_2$
		we have
		that 
		$\Gamma \vdash N_2: A\imp B$
		(by subject reduction).
		Then
		\begin{itemize}
			\item 
			either
			$N_2$ is not an abstraction and
			we conclude by letting 
			$N= \lambda x . N_2 x$.
			\item 
			otherwise
			$N_2=\lambda y.M'$
			and we conclude 
			since 
			$N_1\ered[1] \lambda x. N_2 x\bred[1] N_2$.
		\end{itemize}
		
		\item
		if $M\ered[2] N_1$ 
		with 
		$ \Gamma \vdash M  : \lbox A $ 
		and 
		$ N_1 = \lambdabox x M x$,
		then we conclude with a similar argument with respect to the previous point by letting 
		$N= \lambdabox x {N_2 }x $.

		\item
		if $M\kred N_1$,
		$
		M = \lambdabox{M'}{\myvec{P}, \rclr{N}, \myvec{Q}}{\myvec{x}, \rclr{x} , \myvec{y}} \kred[1]
		\lambdabox{M'}{\myvec{P}, \myvec{Q}}{\myvec{x},\myvec{y}} =N_1 $,
		or
		$M=\lambdabox{M'}{\myvec P,\rclr{N,N}, \myvec Q}{\myvec x,\rclr{y_1}, \rclr{y_2},\myvec{z}}\kred[2]\lambdabox{ M \subst{y, y}{y_1, y_2} }{\myvec P,\rclr{N},\myvec Q 	}{	\myvec x,\rclr{y},\myvec{z}}=N_1$.
		In both cases we conclude with an argument similar to the one in Case~(\ref{localc:b2}).
		
	\end{enumerate}
\end{proof}

In order to prove the termination of $\bekred$, we define the following measures.

\def\estSet{\mathsf{Est}}
\def\nonLinofin#1#2{|#2|_#1^{\mathsf{non-lin}}}
\begin{definition}\label{def:measEta}
	Let $M$ be a \mlt. 
	We define the following multisets of derivable type assignments:
	$$
	\adjustbox{max width=\textwidth}{$
		\begin{array}{l}
			\estSet_1 (M) 
			= 
			\Set{
				B\imp C
				\mid 
				P \in \subterm M\setminus\amltset
				\text{ such that } 
				M \neq PQ 
				\mbox{ and }
				\Gamma \vdash P:	B\imp C
			} 
			\\
			\estSet_2 (M) 
			= 
			\Set{ 
				\lbox B
				\mid 
				P \in \subterm M\setminus\smltset
				\text{ such that } 
				M \neq \lambdabox{Q }{ \myvec{N}_1, P, \myvec{N}_2  }{\myvec{x}_1, x, \myvec{x}_2}
				\mbox{ and }
				\Gamma \vdash P: \lbox B
			}
		\end{array}
		$}
	$$
	We then define $\etam{M}\coloneqq\etao{M}+\etat{M}$
	with 
	$$
	\etao{M} \coloneqq \sum_{ A \in \estSet_1 (M) } \sizeo{ A } 
	\quand
	\etat{M} \coloneqq \sum_{A  \in \estSet_2(M)  } \sizet{  A }
	$$
	$$
	\mbox{where}\qquad
	\begin{array}{l@{\quad }l@{\quad }l}
		\sizeo{a} = 0 
		&
		\sizeo{ A \imp B } = \sizeo{A} + \sizeo{B} + 1              
		&
		\sizeo{\lbox A} = \sizeo{ A}
		\\
		\sizet{a} = 0 
		&
		\sizet{ A \imp B } = \sizet{A} + \sizet{B}   
		&
		\sizet{\lbox A} = \sizet{ A} + 1 
	\end{array}
	$$ 
	We also define $\kappam M$ as the size of substitution subterms of $M$ as follows:
	$$
	\begin{array}{c}
		\kappam{x} = 0 
		\qquad 
		\kappam{ \lambda x M } = \kappam{M} 
		\qquad 
		\kappam{MN} = \kappam{M} + \kappam{N}
		\\
		\kappam{\lambdabox{M}{N_1, \dots, N_n}{x_1,\dots, x_n}}= \kappam{M}+\kappam {N}+ n
	\end{array}
	$$
\end{definition}
\begin{example}
	Intuitively, the measure
	$\etam\cdot$
	does not take into account all the subterms of $M$, 
	but only the ones on which we can apply the restricted 
	$\ered$.
	For an example, consider the \mlt $M  =  (\lambda z^{ a \to a }. z) y$
	with
	$\etam M=3$ because all four subterms of $M$ are of type a $\imp$-formula, but the subterm $\lambda z.z$ is an abstraction, therefore no $\ered$ can be applied on it.
	If $M\ered N$, because of the restrictions on $\ered$, 
	we have that
	\begin{itemize}
		\item 
		either
		$N = (\lambda z. z) (\lambda v . yv )$
		with
		$\etam N=2$
		because
		no $\ered$ can be applied to the subterms $y$ 
		and 
		$\lambda z.z$ (they occur on the left of an application)
		or
		$\lambda v.yv$ (it is an abstraction),
		but only to the subterms
		$z$ and the whole term $N$;

		\item 
		or
		$N= \lambda v^a . ((\lambda z. z)y) v $
		with
		$\etam N=2$
		because
		$\ered$ can only be applied to the subterms
		$z$ and $y$.
		
	\end{itemize}
\end{example}

\begin{restatable}{lemma}{lemEtaDecrease}\label{lem:etadec}
	Let $M$ and $N$ be \mlts. 
	If $ M \ered N $, 
	either $\etam N<\etam M$ or there is $ N' $ such that $ N \ered N' $ and $\etam{N'}<\etam{M}$.
\end{restatable}
\begin{proof}
	We only discuss the two following cases, since the others are direct consequence of the definitions of $\sizeo\cdot$ and $\sizet\cdot$.
	\begin{enumerate}
		\item 
		If $ \Gamma \vdash M : C  $  with $C=A \imp B$ and $ M \eredo N$, 
		then $ N = \lambda x . M x$.
		Therefore, $\sizeo C=\sizeo {A\imp B}>0$ and 
		$$
		\etao{M}
		=
		\left( \sum_{ C' \in \estSet_1 (M) } \sizeo{C'} \right)
		>
		\sum_{
			C' \in \estSet_1 (M) 
			\setminus 
			\set{ C}
		} 
		\sizeo{C'}
		>
		\etao{\lambda x . M x } 
		=  
		\etao N
		$$ 
		
		Now consider $\etat{N}$. 
		We reason by cases on the structure of the type $ A$. 
		If $ A $ is not a box, we can conclude by setting $ N' = N$.  
		If $ A $ is of the shape $ \lbox A'$ then $\etat{N } = \sum_{C' \in \estSet_2 (M) \cup\set{A}} \sizet{C'}$.
		This means that we can perform a step $ \lambda x . M x \eredt \lambda x. (M (\lambdabox{z}{x}{z}))   $ with $z $ fresh. 
		Now we set $ N' =  \lambda x. (M (\lambdabox{z}{x}{z}))$ since 
		$\etat{N'} =  \sum_{C' \in \estSet_2 (M) } \sizet{C'}  = \etat{M}$. 
		We remark that $\etao{N'} = \etao{N}$. 
		Hence $ \etam{N'} = \etao{N} + \etat{M} <  \etao{M} + \etat{M} = \etam{M}   $ since $ \etao{N} < \etao{M} $ for what we said before.

		\item 
		If $ \Gamma \vdash M : C$  with $C=\lbox A$ and $ M \eredt N$, 
		then $  N = \lambdabox{x}{M}{x}$.
		Therefore, $\sizet C=\sizet{\lbox A}>0$ and 
		$$
		\etat{M}
		=
		\left( \sum_{ C' \in \estSet_1 (M) } \sizet{C'} \right)
		>
		\sum_{
			C' \in \estSet_1 (M) 
			\setminus 
			\set{ C}
		} 
		\sizeo{C'}
		>
		\etat{\lambda x . M x } 
		=  
		\etat{ \lambdabox{x}{M}{x} } 
		= 
		\etat N
		$$
		
		Now consider $ \etao{N}$.  
		We reason by cases on the structure of the type $ A$. 
		If $ A $ is not an implication, we can conclude by setting $ N' = N$.  
		If $ A $ is of the shape $ A' \imp B$ then $\etao{N } = \sum_{C' \in \estSet_1 (M) \cup \set{A}} 
		\sizet{C'}$. 
		This means that we can  perform a step  $ \lambdabox{x}{M}{x} \eredt   \lambdabox{(\lambda y. x y)}{M}{x}  $ with $y $ fresh. 
		Now we set $ N' = \lambdabox{(\lambda y. x y)}{M}{x}  $ since $  \etao{N'} =  \sum_{
			C' \in \estSet_1 (M) } \sizet{C'}  = \etao{M}$. 
		We remark that $\etat{N'} = \etat{N}$.
		Hence $ \etam{N'} = \etat{N} + \etao{M}<\etat{M} + \etao{M} = \etam{M}$ since $ \etat{N} < \etat{M}$ for what we said before. 
		
	\end{enumerate}
\end{proof}

\begin{lemma}\label{lem:subcommeta}
	Let $ P, Q, N  $ be \mlts.
	If $ P \subst Q x \ered N $, then there are $ N_1$ and $N_2$ such that $  P \ered^\ast N_1$ and $Q \ered^\ast N_2 $ with  $ N_1 \subst  {N_2} {x} = N $.
\end{lemma}
\begin{proof}
	We prove it by induction on the structure of $ P$.
	\begin{itemize} 
		\item If $ P = x $ then $P \subst{Q}{x}=Q \ered N$.
		Then $  N_1 = x $ and $ N_2 = N$.
		
		\item 
		If $ P = \lambda y. P' $, then $  P \subst Q x = \lambda y . P' \subst Q x$.
		By definition, if $   P \subst Q x \ered N $ then 
		$ \Gamma \vdash P : A\imp B  $ and then 
		$N = \lambda z . P''     $ with $ P' \subst Q x \ered P''$.
		In this case we apply inductive hypothesis to get $ N'_1  $ and $ N'_2 $ such that $P' \ered N'_1 $ and $ Q \ered N'_2  $ with $ {N'_1} \subst {N'_2} {x} = P'' $.
		We conclude by letting   $ N_1 = \lambda y. N'_1 $ and $  N_2 = N'_2 $;

		\item 
		If $ P = ST$, then $   P \subst Q x = S \subst {Q} {x} T \subst {Q} {x}$.
		By definition of $\ered$, if $ P \subst Q x \ered N $ then:
		\begin{itemize}
			\item 
			if $ S \subst {Q} {x} T \subst {Q} {x} \ered[1] \lambda y. ( S \subst {Q} {x} T \subst {Q} {x}) y=N$, 
			then $ N_1 = \lambda y (ST) y$ and $N_2 = Q $;
			
			\item 
			if $  S \subst {Q} {x} T \subst {Q} {x} \ered[2] \lambdabox{ y}  { S \subst {Q} {x} T \subst {Q} {x}} {y}=N$, 
			then $ N_1 =  \lambdabox{y}{ST}{y}$ and $N_2 = Q $;
			
			\item
			otherwise, the step $ N \ered N' $ must be applied in a context. 
			In this case, we have that $ N = P_1 P_2 $ and either $ S \subst Q x \ered P_1 $ or $ T \subst Q x \ered P_2 $. 
			In the case $  S \subst Q x \ered P_1$, by definition of the contextual step, $\subst Q x \ered P_1 $ cannot be a step of $\ered[1]$. 
			By inductive hypothesis there are $ N_{1,2}  $ and $ N_{2,2} $ such that $S \ered N_{1,2} $ and $ Q \ered N_{2,2} $ with $P_1 = {N_{1,1}} \subst {N_{2,2}}{x}$. 
			Then
			$ N_1 = N_{1,1} (T \subst Q x) $ and $ N_2 = N_{2,2}$.
			The other case is similar.
			
		\end{itemize} 
		\item 
		If $ P = \lambdabox{S}{\myvec T}{\myvec y} $ 
		then $   P \subst Q x = \lambdabox{S \subst {Q} {x}}{ {\myvec {T}} \subst {Q} {x}}{\myvec{y}}$.
		By definition of $\ered$, if $ P \subst Q x \ered N $ then:
		\begin{itemize}
			\item 
			if $ \lambdabox{S \subst {Q} {x}}{ {\myvec {T}} \subst {Q} {x}}{\myvec{y}} \ered[1] \lambda y. ( \lambdabox{S \subst {Q} {x}}{ {\myvec {T}} \subst {Q} {x}}{\myvec{y}}) y=N$, 
			then 
			$ N_1 = \lambda y . (\lambdabox{S}{\myvec{T}}{\myvec{y}}) y$ and $N_2 = Q $;

			\item
			otherwise, the step $ N \ered N' $ must be applied in a context. 
			In this case, we have that $ N = \lambdabox{P_1 }{\myvec{P}_2}{\myvec{y}} $ and either $ S \subst Q x \ered P_1 $ or $ \myvec{T} \subst Q x \ered \myvec{P}_2 $. 
			We do the case $  S \subst Q x \ered P_1$.  
			By inductive hypothesis there exists $ N_{1,2}  $ and $ N_{2,2} $ such that 
			$  S \ered N_{1,2} $ and $ Q \ered N_{2,2} $ 
			with $P_1 = {N_{1,1}} \subst {N_{2,2}}{x}$. 
			Then 
			$ N_1 = \lambdabox{N_{1,1} }{\myvec{T} \subst Q x }{\myvec{y}}$ and $ N_2 = N_{2,2}$.
			The other case is similar.	
		\end{itemize}
	\end{itemize}
\end{proof}

\def\diag#1#2#3#4#5{
	if 
	$M\rightsquigarrow_{#1} N \rightsquigarrow_{#3} N'$,
	then there is $ M' $ such that
	$M\rightsquigarrow_{#2} M'$
	and
	$M'\rightsquigarrow_{#4}^{#5} N'$
	\xspace
}
\begin{restatable}{lemma}{lemNormpr}\label{lem:normpr}
	The following commutations between $\bred$, $\ered$ and $\kred$ hold:
	\begin{itemize}
		\item \diag \kappa\beta\beta\kappa*;
		\item \diag\eta\kappa\kappa\eta*;
		\item \diag\beta\eta\eta\beta*.
	\end{itemize}
\end{restatable}
\begin{proof}
	After \Cref{lemma:substitution,lem:subcommeta},
	we can reason on the structure of $M$ only considering its possible shape according to the ground steps of $\bekred$.
	\begin{itemize}
		
		\item 
		If $  M \ered[1] N = \lambda x. M x$, 
		then, by \cref{lem:subcommeta} we have $ N_1$ and $ N_2 $ such that $ P \ered N_1$ and $Q \ered N_2 $ with $  {N_1} \subst {N_2} x  =  N$.
		We conclude by letting $M' =  (\lambda x . N_1) N_2$.
		A similar argument is applied if $M\ered[2]N$.

		\item 
		If $ M = (\lambda x . P)Q  \bred[1] P \subst  Q x  =N$, then we have $  N_1$ and $ N_2  $ such that $ P \ered N_1$ and $Q \ered N_2 $ with $  {N_1} \subst {N_2} x  =  N$.
		We conclude by letting $  M' =  (\lambda x . N_1) N_2$.
		A similar argument is applied if $M \bred[2] N$.
		
		\item 
		If $ M = \lambdabox{M}{ \myvec{P}, N , \myvec{Q}  }{\myvec{x} , y, \myvec{y}  } \kred[1]  \lambdabox{M}{ \myvec{P}, \myvec{Q}  }{\myvec{x} ,  \myvec{y}  } =N $, 
		then we can conclude since any $\beta $-redex of $ N $ is also a $\beta$-redex of $  M$. 
		A similar argument is applied if $M\kred[2]N$.
		
	\end{itemize}
\end{proof}

Therefore, the following corollary trivially holds.
\begin{corollary}
	The following hold.
	\begin{itemize} 
		\item\label{prop:normpr:2} 
		if $ M \kred N  $,
		then $ \nf{M}{\beta} \kred^\ast \nf{ N}{\beta}$;
		
		\item\label{prop:normpr:3} 
		if $ M \ered N $,
		then $ \nf{ M}{\kappa} \ered \nf{N}{\kappa}$;
		
		\item\label{prop:normpr:4} 
		if $ M \bred N  $,
		then $\nf{ M}{\eta} \bred^\ast \nf{N}{\eta}$.
	\end{itemize}
\end{corollary}

\begin{restatable}{theorem}{thmSNadnC}\label{thm:SNandC}
	The reduction relation $\bekred$ is strongly normalizing and confluent.
\end{restatable}
\begin{proof}
	After \Cref{prop:locConf}, 
	it suffices to prove that $\bekred$ is strongly normalizing to conclude by Newman's lemma that $\bekred$ is also confluent.

	To prove strong normalization we use the fact that 
	the reductions
	$\bred$, $\ered$ and $\kred$
	are strongly normalizing:
	for $\bred$ the proof can be found in \cite{kakutanilambda},
	for $\ered$ the proof is by induction on $\etam\cdot$ using \Cref{lem:etadec}, and
	for $\kred$ it follows the fact that, by definition of $\kappam \cdot$, we have that $\kappam M>\kappam N$  whenever $ M \kred N $.
	To conclude that $\bekred$ also is strongly normalizing,
	the standard result (see, e.g., \cite{terese}) in rewriting theory ensuring that
	given
	two strongly normalizing reduction relations 
	$\rightsquigarrow_1$ and $\rightsquigarrow_2 $ 
	with $ \rightsquigarrow_1 $ confluent,
	if 
	$ M \rightsquigarrow_2 N $ implies the existence of a reduction $\mathsf{nf}_1(M) \transc{2} \mathsf{nf}_1 (N)$ for any $M $ and $N$, ,
	then 
	$\rightsquigarrow_1 \cup \rightsquigarrow_2$ is strongly normalizing.
	In our case,
	the fact that 
	$ M \rightsquigarrow_2 N $ implies $\mathsf{nf}_1(M) \transc{2} \mathsf{nf}_1 (N)$ 
	is a corollary of \Cref{lem:normpr}.
\end{proof}

\begin{definition}\label{def:normTerms}
	The set $\normTerms$ 
	is the set of \mlts defined inductively as follows:
	\begin{itemize}
		
		\item 
		if 
		$x$ is a variable, $T_1,\ldots, T_n\in\normTerms$, 
		and 
		there are derivations for the types assignments 
		$ \Gamma \vdash x : (A_1,\ldots, A_n) \to C$  with $ C$ atomic
		and
		$ \Gamma \vdash T_i: A_i$  for all $i\in\intset1n$,
		then $xT_1\cdots T_n\in \normTerms$.
		Variables are the special case with $n=0$;

		\item 
		if 
		$ T \in \normTerms$ 
		and
		there is a derivation of 
		$ \Gamma, x : A \vdash T:C $,
		then 
		$\lambda x^A . T\in \normTerms$;
		
		\item 
		if 
		$M\in \normTerms$, 
		$\FV M=\set{x_1, \ldots, x_n}$ and
		the type assignment
		$x_1 : B_1, \dots, x_n : B_n \vdash M : C $ is derivable,
		and
		if
		there are $n$ distinct terms
		$T_1,\ldots,T_n\in\mltset$ 
		of the shape 
		$T_i=y_iU_{i1}\cdots U_{ik_i}$ 
		with 
		$U_{ij}\in\normTerms$ for all $i\in\intset1n$ and $j\in\intset1{k_i}$,
		such that
		the type assignment $ \Gamma \vdash T_i:\lbox B_i$ is derivable for all $i\in\intset1n$,
		then
		$\lambdabox{M}{T_1,\ldots,T_n}{x_1,\ldots,x_n}\in \normTerms$.
		
	\end{itemize}
\end{definition}
%
%
%
%
\begin{restatable}{proposition}{propNormal}\label{prop:normal}
	The set $\normTerms$ is the set of \mlts in $\beta\eta\kappa$-normal form $\nfset$.
\end{restatable}
\begin{proof}
	By definition, every $\normTerms\subseteq\nfset $ is  $\bekred$-normal.
	To prove the converse we proceed by induction on the structure of $ M \in\nfset$:
	\begin{itemize}
		\item 
		if $ M  = x$,
		then $M\in\normTerms$ by definition;
		
		\item 
		if $ M = \lambda x . M'\in \nfset$, then also $ M' \in \nfset$.
		By inductive hypothesis, this implies $ M' \in \normTerms$.
		Therefore $ \lambda x . M' \in \normTerms$;
		
		\item 
		if $ M = PQ \in\nfset$, then both $ P$ and $ Q$ are in $\nfset$
		and
		there is a derivable type assignment $ \Gamma \vdash M : C$, 
		and derivable type assignments
		$ \Gamma\vdash P : A \imp C$ and $ \Gamma \vdash Q : A$.
		We have that no $\ered$-rule can be applied to $C$ because $M\in\nf\Lambda\eta$; thus $C$ must be atomic.
		We know that $P$ cannot be in $\amltset$ since $M\in\nf\Lambda\beta$ and $P$ cannot be in $\smltset$ because $ \Gamma\vdash P : A\imp  C$ is derivable.
		Then by inductive hypothesis we have that $ P = x T_1, \dots T_n $  for some $T_1,\ldots,T_n\in\normTerms$.
		We conclude that $ PQ \in \normTerms$;
		
		\item
		if $ M = \lambdabox{P}{Q_1, \dots, Q_n}{x_1, \dots, x_n} \in\nfset$,
		then 
		there is a derivable type assignment
		$ x_1 : B_1, \dots, x_n : B_n \vdash P : C $
		and derivable type assignments $\Gamma \vdash Q_i:\lbox B_i$ for all $i\in\intset1n$.
		Since $M\in\nfset$, then no $\bekred$-rule can be applied to $M$, nor to $P$; thus $P\in\nfset$.
		Similarly, since $M\in\nfset$,
		then 
		$Q_i\notin\smltset$ (otherwise we could apply $\bred^2$),
		$Q_i\in\nf{\Lambda}{\beta\kappa}$ (since no $\bkred$-rule can be applied to $ Q_i$)
		and
		$Q_i$ cannot be in $\nf\Lambda\eta$ (because $Q_i:\lbox B_i$ and otherwise $\ered$-steps could be applied on $M$)
		for all $i\in\intset1n$.
		We conclude that $M\in\normTerms$.
	\end{itemize}
\end{proof}

\section{A Canonical Type System for $\CK$}\label{sec:foc}

\def\aximp{\mathsf{\imp_L^{ax}}}
\def\kimp{\mathsf{\imp_L^K}}
\def\mlimp{\mathsf{\imp_R^*}}

\begin{figure}[t]
	\adjustbox{max width=\textwidth}{$\begin{array}{c}
			\vlinf{\taxrule}{}{\Gamma,x:c\vdash x:c}{}
			\qquad
			\vlinf{\exrule}{\ast}{\sigma(\Gamma)\vdash M:C}{\Gamma\vdash M:C}
			\qquad
			\vlinf{\kbrule}{\star}{
				\Delta, y_1 :   \lbox A_1,\ldots, y_n  :  \lbox A_n \vdash \lambdabox{M}{x_1,\ldots, x_n}{y_1,\ldots,y_n}  :  \lbox  C
			}{x_1 :  A_1,\ldots,x_n :  A_n \vdash M :   C}
			\\\\
			\vlinf{\aximp}{\S}{
				\Gamma, y :  \underbrace{(A_1,\ldots,  A_n)\imp c}_B\vdash yN_1\cdots N_n  :  c
			}{
				\left\{
				\Gamma,y :  B \vdash N_i :  A_i
				\right\}_{i\in\intset1n}
			}
			\qquad
			\vlinf{\mlimp}{}{
				\Gamma\vdash \lambda x_1^{A_1}\cdots x_n^{A_n}. M :  (A_1,\ldots A_n) \imp C 
			}{\Gamma ,x_1 :  A_1, \ldots,x_n :  A_n \vdash M :  C}
			\\\\
			\vliinf{\kimp}{\dagger, \S}{
				\Gamma, 
				\underbrace{
					f_1 :  (A_{1,1},\ldots, A_{1,k_1})\imp \lbox B_1
					,\ldots, 
					f_n  :  (A_{n,1},\ldots, A_{n,k_n}) \imp \lbox B_n
				}_{\Delta}
				\vdash 
				\lambdabox{M}{N_1,\ldots, N_n,\myvec{z}}{y_1,\ldots y_n,\myvec{w}}  :  \lbox C
			}{
				\left\{
				\Gamma,\Delta \vdash T_{i,j}  :   A_{i,j}
				\right\}_{
					i\in\intset1n,
					j\in\intset1{k_i}
				}
			}{	
				\Gamma, \Delta, x_1  :  \lbox  B_1, \ldots,  x_n  :  \lbox B_n \vdash \lambdabox{M}{x_1,\ldots,x_n, \myvec{z}}{y_1,\ldots , y_n,\myvec{w}}  :  \lbox C
			}
			\\
			\begin{array}{l@{\;\coloneqq\;}l@{\qquad\qquad}l@{\;\coloneqq\;}l}
				\ast
				&
				\sigma \mbox{ permutation over }\intset1n
				&
				\star 
				& 
				\mbox{ $FV(M)=\set{x_1,\ldots x_n}$ and $y_1,\ldots, y_n$ fresh}
				\\
				\S
				&
				\mbox{each $N_i = f_i T_{i,1} \cdots  T_{i,k_i}$ for $i\in \set{1,\ldots, n}$}
				&
				\dagger
				&
				\mbox{$\Gamma$ contains no formula of the shape $(A_1\cdots A_n )\imp \lbox B$}
			\end{array}
		\end{array}$}
	\caption{Typing rules of the typing system $\FCK$.}
	\label{fig:focus}
\end{figure}

In this section we present an alternative typing system for modal $\lambda$-terms
where each term in $\normTerms$ admits exactly one typing derivation. 
The rules of this system (we call $\FCK$) are provided in \Cref{fig:focus} and are conceived to reduce the non-determinism of the typing process,
following the same approach used in designing focused sequent calculi~\cite{andreoli:01,mil:vol:foc,cha:mar:str:fscd16}.
Derivations and derivability in $\FCK$ are defined analogously to \Cref{def:derivation}, using rules in $\FCK$ instead of rules in $\NDCK$. We remark that the structural rules of weakening and contraction are admissible in the system.

We can now prove a result of \emph{canonicity} of $\FCK$ with respect to typing derivations of \mlts in $\nfset$.

\begin{restatable}{theorem}{NDisFCK}\label{thm:DNisFCK}
	Let 
	$T\in \normTerms$
	and 
	$\Gamma \vdash T:A$
	be a derivable \as.
	Then there is a unique (up to $\exrule$-rules) derivation of $\Gamma \vdash T :  A$ in $\FCK$.
\end{restatable}
\begin{proof}
	The proof of this theorem follows from the 
	correspondence between the inductive definition of terms in $\normTerms$ (\cref{def:normTerms})
	and the shape of the typing rules of $\FCK$.
		
	By definition of $\normTerms$, 
	we have the following cases:
	\begin{itemize}
		\item 
		if $M=x$ is a variable, then $A=a$ must be an atomic formula. 
		Since the sequent $\Gamma \vdash x :  a$ is derivable in $\NDCK$, 
		then it can only be the conclusion of a $\Idrule$-rule if $x :  a \in \Gamma$.
		We conclude since the rule $\Idrule$ is also in $\FCK$ and such a derivation is unique;
		
		\item 
		if $T= x N_1 \cdots N_m$,
		then $x: C=(A_1,\ldots,  A_n) \to A$ is in $\Gamma$
		and 
		$N_1,\ldots, N_m\in \normTerms$ with size smaller than $\sizeof M$.
		By inductive hypothesis,
		for each $i\in\intset1m$ the sequent $\Gamma \vdash N_i : A_i$ 
		is derivable in $\NDCK$, therefore admits a unique derivation $\dD_i$ in $\FCK$.
		We conclude since we have a unique derivation of $\Gamma \vdash M :  A$ 
		starting with a by $\aximp$  
		whose premises are the conclusions $\dD_1, \ldots, \dD_n$;

		\item 
		if $T= \lambda x_1^{A_1}\cdots\lambda x_n^{A_n}.M $ for a $M\neq \lambda y.N$,
		then $A=(A_1,\ldots, A_n)\imp C$ for some types $A_1,\ldots, A_n,C$.
		Applying $n$ times the rule $\impintro$ we know that 
		$\Gamma \vdash  \lambda x^{A_1}\ldots\lambda x^{A_n}.M : (A_1,\ldots, A_n)\imp C$ 
		iff 
		$\Gamma, x : A_1, \ldots, x_n : A_n \vdash M : C$.
		By induction, we know that there is a unique derivation $\dD_1$ of the sequent $\Gamma, x : A_1, \ldots, x_n : A_n \vdash M : C$ in $\FCK$.
		We conclude since we have a unique derivation of $\Gamma \vdash M :  A$ 
		starting with a by $\mlimp$  
		whose premise is the conclusion of $\dD_1$;

		\item 
		if $T= \lambdabox{M}{N_1,\ldots ,N_n}{y_1,\ldots,y_n} $,
		then, we have two cases:
		
		\begin{itemize}
			\item  
			$N_i=x_i$ is a variable for all $i\in\intset1n$. 
			In this case a $\NDCK$ derivation of $\Gamma \vdash x_i : \lbox A_i$ can only be made of a single $\Idrule$-rule. 
			This implies that 
			$x_i : \lbox A_i\in \Gamma$  for all $i\in\intset1n$;
			thus  
			$\Gamma= \Delta, x_1: \lbox A_1,\ldots x_n : \lbox  A_n$ 
			for a context $\Delta$.  
			Moreover we must have a $\NDCK$-derivation of 
			$x_1: A_1,\ldots, x_n : A_n \vdash M: C$ 
			for a $C$ is such that $\lbox C= A$;
			thus, since $\sizeof M<\sizeof T$, there is a unique derivation $\dD_1$ of this latter sequent in $\FCK$.  
			We conclude since we have a unique derivation of $\Gamma \vdash M :  A$ 
			starting with a by $\kbrule$  
			whose premise is the unique derivation of 
			$\Delta, x : A_1, \ldots, x_n : A_n \vdash M : C$ in $\FCK$;

			\item 
			there are some $i\in\intset1n$ such that $N_i$ is of the form $f_iT_{i,1}\cdots T_{i,k_i}$ with 
			$f_i: (A_{i,1}, \ldots,A_{i,k_i})\imp \lbox B_i$
			and 
			$T_{i,1},\ldots, T_{i,k_i}\in\normTerms$.
			For any $i\in\intset1n$ and $j\in\intset1{k_i}$ the sequent $\Gamma \vdash T_{i,j}:A_{i,j}$ is derivable in $\NDCK$, 
			then,
			since the weakening rule is admissible,
			also 
			$\Gamma, f_1: (A_{1,1}, \ldots,A_{1,k_1})\imp \lbox B_1,\ldots, f_n: (A_{n,1}, \ldots,A_{n,k_n})\imp \lbox B_n \vdash N_{i,j}: A_{ij}$ is derivable.
			Since  $\sizeof {N_{i,j}}<\sizeof M$,
			we can conclude by induction the existence of a unique derivation $\dD_{i,j}$  in $\FCK$ for this latter sequent.
			By similar argument, there is also a unique derivation for 
			$\Gamma, \Delta, x_1  :  \lbox  B_1, \ldots,  x_n  :  \lbox B_n \vdash \lambdabox{M}{x_1,\ldots,x_n, \myvec{z}}{y_1,\ldots , y_n,\myvec{w}}  :  \lbox C$
			allowing us to conclude the existence of a 
			unique derivation of $\Gamma \vdash M :  A$ 
			starting with a by $\kimp$  
			whose
			rightmost premise is
			the conclusion of $\dD'$
			and 
			whose other premises are the derivations $\dD_{i,j}$ with $i\in\intset1n$ and $j\in\intset1{k_i}$;
		\end{itemize}
	\end{itemize}
\end{proof}

\section{Game Semantics for $\CK$}\label{sec:games}

In this section we recall definitions and results on the winning innocent strategies for the logic $\CK$ defined in \cite{acc:cat:str:games}.
For this purpose, we first recall the construction extending Hyland-Ong arenas \cite{hylandPCF,Murawski:Ong} for intuitionistic propositional formulas  to represent formulas containing modalities, and then we recall the characterization of the winning innocent strategies representing proofs in $\CK$.
We conclude by proving the full-completeness result between for those strategies by showing a one-to-one correspondence between strategies for type assignments of terms in normal forms and their (unique) typing derivations in $\FCK$.

\subsection{Arenas with Modalities}
\def\root{sink\xspace}
\def\roots{sinks\xspace}

We recall the definition of arenas with modalities from \cite{acc:cat:str:games} extending the encoding of arenas from \cite{hyland:ong:00,ICP}.
For this purpose, we assume the reader familiar with the definition of 
\defn{two-color directed graph} (or \defn{\ddags} for short),
i.e., 
directed acyclic graphs with two disjoint sets of directed edges $\iedge$ and $\medge$ (details can be found in \cite{acc:cat:str:games,ICP}).

\begin{definition}\label{def:arena}
	The \defn{arena}
	of a formula $F$ is the \ddag $\arof F$
	with  vertices are labeled by elements in $\labelset =\atoms \cup \set{\lbox}$ inductively defined as follows:
	\begin{equation}\label{eq:translation}
		\adjustbox{max width=.9\textwidth}{$
			\begin{array}{c}
				\arof a = \singlevertex[a]
				\qquad
				\arof{A \imp B} = \arof A \gimp \arof B
				\qquad
				\arof{\lbox A}	= {\ssbox}\;  \gmod \arof A

			\end{array}
		$}
	\end{equation}
	where
	$\singlevertex a$ and $\ssbox$ denote the graphs consisting of a single vertex labeled by $a$ and $\lbox$ respectively,
	and
	where
	the binary operation $\gimp$ and $\gmod$ on \ddags
	are defined as follows:
%
\begin{equation*}\adjustbox{max width=\textwidth}{$
		\begin{array}{l@{\;=\;}l@{\;}l@{\:,\:}l@{\:,\:}l@{\;}r@{}l}
			\gG \gimp \gH
			&
			\tuple{
				&
				\vertices[\gG]\uplus \vertices[\gH]
				& 
				\iedge[\gG\uplus\gH] \cup
				\RtoR\gG\gH
				&
				\medge[\gG\uplus\gH]
				&
			}
		\end{array}
		\quand
		\begin{array}{l@{\;=\;}l@{\;}l@{\:,\:}l@{\:,\:}l@{\;}r@{}l}
			\gG \gmod \gH
			&
			\tuple{
				&
				\vertices[\gG]\uplus \vertices[\gH]
				&
				\iedge[\gG\uplus\gH]
				&
				\medge[\gG\uplus\gH]\cup
				\RtoR\gG\gH
				&
			}
		\end{array}
	\quad\mbox{with}
	$}
\end{equation*}
\begin{equation*}\adjustbox{max width=\textwidth}{$
	\begin{array}{c@{\;=\;}l}
		\vertices[\gG]\uplus \vertices[\gH]
		&
		\Set{(v_i,i)\mid i\in\set{0,1} \mbox{ and } v_0\in \vertices[\gG] \mbox{ and } v_1\in \vertices[\gH]} 
		\quand
		\lab{(v_i,i)}=\lab{v_i}
		\\
		\dedge[\gG\uplus \gH]
		&
		\Set{\left((v_i,i),(w_i,i)\right)\mid i\in\set{0,1} \mbox{ and } (v_0,w_0)\in\dedge[\gG] \mbox{ and } (v_1,w_1)\in\dedge[\gH] } 
		\quad
		\mbox{for each } \dedge\in\set{\iedge,\medge}
		\\
		\RtoR{\gG}{\gH}
		&
		\Setof{((v,0),(w,1))}{v\in \irof[\gG],w\in \irof[\gH]}
		\quad \mbox{where}\quad 
		\irof[X]\coloneqq\set{v\in\vertices[X]\mid v\iedge[X] w \mbox{ for no }w\in\vertices[X]}
	\end{array}
$}\end{equation*}
%
	The \defn{arena of a sequent}
	$A_1,\ldots,A_n\vdash C$ 
	is the arena $\are$ of $\arof{ (A_1,\ldots,A_n) \imp C}$.
\end{definition}

\begin{remark}\label{rem:roots}
	By construction, an arena $\gG$ of a formula or a sequent $\Gamma \vdash C$
	always admits a unique non $\lbox$-labeled vertex in $\irof[\gG]$,
	i.e., 
	a unique vertex  $v$ with $\lab v \neq\lbox$ such that there is no $w\in\vertices[\gG]$ such that $v\iedge[\gG]w$.
\end{remark}

We draw \ddags by representing a vertex $v$ by its label $\lab v$.
If $v$ and $w$ are vertices of an \ddag, then we draw 
$\vvv1 \quad \vw1 \Aedges{vv1/w1}$ if $v\iedge w$ and 
$\vvv1 \quad \vw1 \Medges{vv1/w1}$ if $v\medge w$. By means of example, consider the arena below.
\begin{equation}\label{eq:arena}
	\arof{\big(a\imp\lbox (b\imp (c\imp \lbox d))\big)\imp \lbox (e\imp f)}
	\quad=\quad
	\def\myskip{\hskip1em}
	\begin{array}{c@{\quad\myskip}c@{\myskip}c@{\myskip}c@{\myskip}c@{\myskip}c@{\myskip}c@{\myskip}c}
		\va1	&\vlbox0&		&		&		&	&	&\vlbox2	\\[2pt]
		&		&		&		&		&\vlbox1& 	& 				\\[2pt]
		&		&\vb1	&		&		&		& \ve1	&\vf1 		\\[2pt]
		&		&		&\vc1	&		&\vd1	& 
	\end{array}
	\Medges{lbox0/lbox1,lbox2/f1,lbox0/d1,lbox1/d1}
	\multiAedges{c1}{d1,lbox1}
	\multiAedges{a1}{lbox1,lbox0,d1}
	\multiAedges{b1}{lbox1,d1}
	\multiAedges{lbox1,d1,lbox0}{lbox2,f1}
	\Aedges{e1/f1}
\end{equation}

\begin{remark}
	All arenas of the form $\arof{ (A_{\sigma(1)},\ldots,A_{\sigma(n)}) \imp C}$ have the same representation for any $\sigma$ permutation over $\intset1n$.
	More in general, it can be shown that the arena of any two equivalent formulas modulo Currying $A\imp (B\imp C)\sim B\imp(A\imp C)$ can be depicted by the same arena.
	However,
	whenever there may be ambiguity because of the presence of two vertices with the same label,
	we may represent the vertex 
	$v=((\cdots(v',i_1) \cdots ),i_n)$
	(where $i_1,\ldots, i_n\in\set{0,1}$) 
	by 
	$\lab v_{i_1, \ldots, i_n}$ instead of simply $\lab v=\lab{v'}$ (see \Cref{ex:arena}).
\end{remark}

\begin{definition}
	Let $\arof F$ be an arena and $v$ one of its vertices.
	The \defn{depth} of $v$ is the number $\dep{v}$ of vertices in a $\iedge$-path from $v$ to a vertex in $\irof[\arof F]$
	\footnote{
		As proven in \cite{ICP,acc:str:AIML22}, 
		arenas are \emph{stratified}, that is, all the $\iedge$-path from a vertex $v$ to any vertex in $\irof[\arof F]$ have the same length.
		Therefore the number $\dep v$ is well-defined.
	}%
	.
	The \defn{address} of $v$ is defined as the unique sequence of modal vertices $\add v=m_1,\dots , m_h$ in $\vertices[\arof F]$ 
	corresponding to the sequence of modalities in the path in the formula tree of $F$ connecting the node of $v$ to the root.
	If $\add v=m_1,\dots , m_h$, 
	we denote by $\addi k {v}=m_k$ its $k^{\mbox{\small th}}$ element
	and
	we call the \emph{height} of $v$ (denoted $\addsize{v}$) the number of elements in $\add v$. 
\end{definition}

\newvertex{lboxz}{\lbox_{011110}}{}
\newvertex{lboxp}{\lbox_{010}}{}
\newvertex{lboxs}{\lbox_{10}}{}
\begin{example}\label{ex:arena}
	Below 
	an alternative representation of its arena of the formula 
	$\big(a\imp\lbox (b\imp (c\imp \lbox d))\big)\imp \lbox (e\imp f)$ in
	\Cref{eq:arena}
	where the ambiguity of the vertex representation is avoided by the use of indices, 
	the corresponding formula-tree,
	and the complete list of the addresses of all vertices in this arena.
	\begin{equation*}
		\adjustbox{max width=\textwidth}{$
			\def\myskip{\hskip1em}
			\newvertex{land}{\land}{}
			\newvertex{imp}{\imp}{}
			\newvertex{unit}{\unit}{}
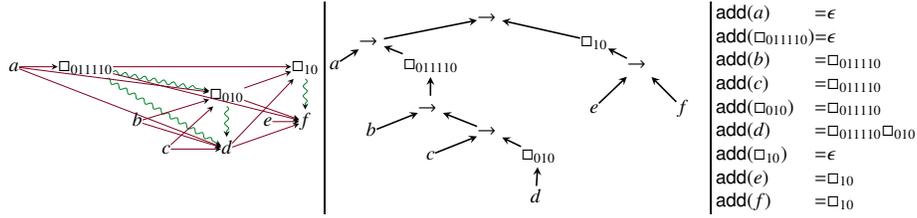
			\begin{array}{c|c|c}
				\begin{array}{c@{\quad\myskip}c@{\myskip}c@{\myskip}c@{\myskip}c@{\myskip}c@{\myskip}c@{\myskip}c}
					\va1	&\vlboxz0&		&		&		&	&	&\vlboxs2	\\[2pt]
					&		&		&		&		&\vlboxp1& 	& 				\\[2pt]
					&		&\vb1	&		&		&		& \ve1	&\vf1 		\\[2pt]
					&		&		&\vc1	&		&\vd1	& 
				\end{array}
				\Medges{lboxz0/lboxp1,lboxs2/f1,lboxz0/d1,lboxp1/d1}
				\multiAedges{c1}{d1,lboxp1}
				\multiAedges{a1}{lboxp1,lboxz0,d1}
				\multiAedges{b1}{lboxp1,d1}
				\multiAedges{lboxp1,d1,lboxz0}{lboxs2,f1}
				\Aedges{e1/f1}
				&
				\begin{array}{c@{\quad}c@{\quad}c@{\quad}c@{\quad}c@{\quad}c@{\;}c@{\quad}c@{\quad}c@{\quad}c@{\quad}c@{\quad}c@{\quad}c}
					&&			&		\vimp0										\\
					&	\vimp2	&&			&  		&		&	\vlboxs1			\\
					\va1	&		&	\vlboxz0	&&		&		&		&\vimp1\;	 		\\[10pt]
					&&\vimp7&		&		&		&\ve1	&		&\vf1				\\
					&\vb1	&		&\vimp3	&			\\
					&&\vc1	&		&\vlboxp2			\\[10pt]
					&&		&		&\vd1\;
				\end{array}
				\Dedges{lboxs1/imp0,imp1/lboxs1,e1/imp1,f1/imp1,imp2/imp0,a1/imp2,lboxz0/imp2,imp7/lboxz0,b1/imp7,imp3/imp7,c1/imp3,lboxp2/imp3,d1/lboxp2}
				&
				\begin{array}{l@{=}l}
					\add a				&\emptyseq			\\
					\add {{\vlboxz0}}	&\emptyseq			\\
					\add b				&{\vlboxz0}			\\
					\add c 				&{\vlboxz0}			\\
					\add{\vlboxp1}		&{\vlboxz0}			\\
					\add d				&{\vlboxz0}\vlboxp1	\\
					\add{\vlboxs2}		&\emptyseq			\\
					\add e				&\vlboxs2			\\
					\add f				&\vlboxs2			
				\end{array}
			\end{array}
		$}
	\end{equation*}
\end{example}

\subsection{Games and Winning Innocent Strategies}

\newcommand{\seq}[1][]{\mathsf s_{#1}}
\def\gameof#1{\mathcal G_{#1}}
\newcommand{\move}[2][]{\mathsf m_{#2}^{#1}}

In this subsection, we briefly recall the definitions of games and winning strategies from \cite{acc:cat:str:games} required to make the paper self-contained.
Note that differently from the previous works, we here include the additional information of the \emph{pointer function} in the definition of views.
This information is crucial for the results in \Cref{sec:foc} where we provide a one-to-one correspondence between our winning strategies and  \mlts.


\begin{definition}
	Let $\are$ be an arena.
	We call a \defn{move} an occurrence of a vertex $v$ of $\are$ with $\lab v\neq \lbox$. 
	The \defn{polarity} of a move $v$ is the parity of $\dep v$:
	a move is a \defn{\evenmove} (resp.~a \defn{\oddmove}) if 
	$\dep v$ is even (resp.~odd). 

	A  \defn{pointed sequence} in $\are$ is a pair $\jpath=\pair\seq f$
	where 
	$\seq=\seq[0],\ldots, \seq[n]$ is a finite sequences of moves in $\are$
	and 
	a \defn{pointer function} $f\colon \intset1{n} \to \intset0{n-1}$ such that $f(i)<i $ and 
	$\seq[i] \iedge[\are] \seq[f(i)] $.
	The \defn{length} of $\jpath$ (denoted $\sizeof{\jpath}$) is defined as the length of $\seq$, that is, $\sizeof\jpath=n+1$.
	Note that we also use $\emptyseq$ to denote the \defn{empty pointed sequence} $\tuple{\emptyseq,\emptyset}$.
\end{definition}

\begin{remark}
	It follows by definition of view that the player $\evensym$ (resp.~$\oddsym$) can only play vertices whose $\dep v$ is even (resp.~odd).
	For this reason, for each $v\in\vertices[\gG]$ we write $\iseven v$ (resp.~$\isodd v$) if the parity of $\dep v$ even (resp.~odd).
	
	Note that the parity of a modality in the address of a move may not be the same as the parity of the move itself.
	By means of example, consider 
	the vertex $c$ in \Cref{ex:arena} which belongs in the scope of two modalities $\vlboxz0$ and $\vlboxp1$ with odd parity.
\end{remark}

Given two pointed sequences 
$\jpath=\pair\seq f$ and $\jpath'= \pair{\seq'}{f'}$ in $\are$, 
we write $\jpath \sqsubseteq \jpath'$ whenever $\seq$ is a prefix of $\seq'$ (thus $\sizeof \seq \leq \sizeof {\seq'}$) and $f(i)=f'(i)$ for all $i\in\intset1{\sizeof{\jpath'}}$
and 
we say that $\jpath$ is a \defn{predecessor} of $\jpath'$ if
$\jpath \sqsubset \jpath'$ and $\sizeof\jpath=\sizeof{\jpath'}-1$.

\begin{definition}
	Let $\are$ be an arena.
	A \defn{play} on $\are$ is a pointed sequence $\jpath=\pair\seq f$ such that, 
	either $\seq=\emptyseq$, 
	or $\seq[i]$ and $\seq[i+1]$ have opposite polarities for all $i\in\intset0{\sizeof{\jpath}-1}$. 

	The \defn{game of $\are$} (denoted $\gameof\are$) is the set of prefix-closed sets of plays over $\are$.

	A \defn{view} is a play $\jpath=\pair\seq f$ such that either $\jpath=\emptyseq$ or the following properties hold:
	
	\hskip1em\begin{tabular}{@{- }l@{ : }l}
		$\jpath$ is \defn{$\evensym$-shortsighted}
		&
		$f(2k)=2k-1$ 
		for every 
		$2k\in\intset2{\sizeof\jpath}$;
		\\
		$\jpath$ is \defn{$\oddsym$-uniform}
		&
		$\lab{\seq[2k+1]}=\lab{\seq[2k]}$ 
		for every 
		$2k+1\in\intset0{\sizeof{\jpath}}$.
	\end{tabular}	

\noindent
	
%
%
%
%

	A \defn{winning innocent strategy} (or \WIS for short) for the game $\gameof\are$ is a finite non-empty prefix-closed set $\strat$ of views in $\gameof\are$ such that:
	
	\hskip1em\begin{tabular}{ll}
		- $\strat$ is \defn{$\evensym$-complete}:
		&
		if $\view\in \strat$ and $\view$ as odd length,
		\\&
		then every successor  of $\view$ (in $\gameof\are$) is also in $\strat$ ;
		\\
		- $\jpath$ is \defn{$\oddsym$-total}:
		&
		if $\view\in \strat$ and $\view$ has even length,
		\\&
		then exactly one successor of $\view$ (in $\gameof\are$) is in $\strat$  ;
	\end{tabular}
	
\noindent A view is \defn{maximal} in $\strat$ if it is not prefix of any other view in $\strat$. $\strat$ is \defn{trivial} 
if  $\strat=\set\emptyseq$. We  say that $\strat$ is a \WIS for a sequent $A_1,\ldots, A_n \vdash C$ 
if $\strat$ is a \WIS for $\arof{A_1,\ldots, A_n \vdash C}$.
	
\end{definition}

The definition of \WIS above is a reformulation of the one in the literature of game semantics for intuitionistic propositional logic~\cite{hylandPCF,dan:heb:reg:game,ICP}.
In presence of modalities, this definition requires to be refined to guarantee the possibility of gather modalities in \emph{batches} corresponding to the modalities introduced by a single application of the $\kbrule$ (see \Cref{fig:seqCK}).
By means of example, consider the following arenas and their corresponding \WISs, which cannot represent valid proofs in $\CK$ because of the impossibility of applying rules handling the modalities in a correct way.

\begin{example}\label{example:2}
	Consider the formulas 
	$F_1=(\lbox a)\imp      a$
	and
	$F_2=(\lbox a \imp \lbox b) \imp \lbox (a \imp  b)$
	and their arenas in \Cref{fig:exampleWIS}.
	The set of views
	$\strat_1$
	and
	$\strat_2$
	are \WISs for $F_1$ and $F_2$ respectively.
	However,
	these formulas are not provable in $\SCK$ because the proof search fails (see  \Cref{fig:exampleWIS}).
	In particular, in the first case, no $\kbrule$ can be applied because only there is a mismatch between the modalities on the left-hand side and on the right-hand side of the sequent;
	in the second case the problem is more subtle and, intuitively, is related to the fact that each $\kbrule$ can remove only a single $\iseven \lbox$ at a time, corresponding to the modality of the unique formula on the right-hand side of the sequent.
\end{example}	

\begin{figure}[t]
	\centering
	\adjustbox{max width = .8\textwidth}{$
	\begin{array}{l|c|c}
		\mbox{Arena}
		&\quad
		\arof{(\lbox a)\imp      a}=
		\begin{array}{c@{\qquad}c}
			\isodd{\vlbox1} 	&  					\\[10pt]
			\isodd{\va1}		& \iseven{\va0}
		\end{array}
		\Aedges{a1/a0,lbox1/a0}
		\Medges{lbox1/a1}
		\quad&\quad
		\arof{(\lbox a \imp \lbox b) \imp \lbox (a \imp  b)}=
		\begin{array}{l@{\qquad}l@{\quad}l@{\quad}l}
			\iseven{\vlbox2_{000}}& \isodd{\vlbox1_{010}} 	&				& \iseven{\vlbox0_{10}} 	\\[10pt]
			&\isodd{\vb1}		&				& \iseven{\vb0}						\\
			\iseven{\va2} 	& 					&\isodd{\va1}	&
		\end{array}
		\multiAedges{lbox1,b1}{b0,lbox0}
		\multiAedges{lbox2,a2}{b1,lbox1}
		\Medges{lbox1/b1,lbox2/a2,lbox0/b0}
		\Aedges{a1/b0}
		\\\hline&&\\
		\mbox{\WIS}
		&
		\strat_1=\Set{\emptyseq,\iseven a, \iseven a \isodd a}
		&
		\strat_2=\Set{\emptyseq,\iseven b, \iseven b \isodd b, \iseven b \isodd b\iseven a, \iseven b\isodd b\iseven a \isodd a }
		\\&&
		\\\hline
		\mbox{\begin{tabular}{c}(failed) \\ Derivation\end{tabular}}
		&
		{\small\vlderivation{\vlin{\rimprule}{}{\vdash \lbox a \imp a }{\vlid{}{}{\lbox a \vdash a}{\vlhy{FAIL}}}}}
		&
		{\small\vlderivation{
			\vlin{\rimprule}{}{\vdash (\iseven{\lbox} a \imp \isodd{\lbox} b) \imp \iseven{\lbox} (a \imp  b)}{
				\vliin{\limprule}{}{\iseven\lbox a \imp \isodd\lbox b \vdash  \iseven\lbox (a \imp  b)}{
					\vlin{\kbrule}{}{\vdash \iseven\lbox a}{\vlid{}{}{\vdash a}{\vlhy{FAIL}}} 
				}{
					\vlin{\kbrule}{}{\isodd\lbox b \vdash \iseven\lbox (a\imp b)}{
						\vlin{\limprule}{}{b \vdash a\imp b}{
							\vlin{\Wrule}{}{b, a\vdash b}{
								\vlin{\AXrule}{}{b\vdash b}{\vlhy{}}
							}
						}
					}
				}
			}
		}}
	\end{array}
	$}
	\caption{Examples of \WISs for arenas not corresponding to proofs.}
	\label{fig:exampleWIS}
\end{figure}

Therefore, in order to capture provability in $\CK$, the notion of winning strategies has to be refined as follows.

\begin{definition}\label{def:CKbatched}
	Let  $\view=(\seq, f)$ be a view in a strategy $\strat$  on an arena $\are$,
	and let $\addsize{\view}=1+\max\set{\addsize{v}\mid v\in\view}$.
	We define the \defn{batched view} of $\view$ as the $\addsize{\view}\times n$  
	matrix 
	$\fviewof{\view}=\big( \fviewof{\view}_{0}, \dots, 	\fviewof{\view}_{n}\big)$ 
	with elements in $\vertices[\gG]\cup\set{\emptyseq}$ such that  
	the 
	each column $\fviewof{\view}_{i}$ is defined as follows:
	\begin{equation*}
			\fviewof{\view}_{i}
			=
			\begin{pmatrix}
				\fviewof{\view}^{\addsize{\view }}_{i}	
				\\
				\vdots
				\\		
				\fviewof{\view}^0_{i}
			\end{pmatrix}
	\qquad\mbox{where}\qquad
	\begin{cases}
		\fviewof{\view}^{\addsize{\view }}_{i}
		=
		\addi{\addsize{\viewi{i}}}{\viewi {i}}
		,\ldots,
		\fviewof{\view}^{h_{\view}-h_{\viewi i}+1}_{i}
		=
		\addi{1}{\viewi {i}} 
		\\
		\fviewof{\view}^{h_{\view}-h_{\viewi i}}_{i}
		=
		\emptyseq
		,\ldots,
		\fviewof{\view}^1_{i}
		=
		\emptyseq
		\\
		\fviewof{\view}^0_{i}
		=
		\viewi i
	\end{cases}
	\end{equation*}
	We say that 
	$\view$ is \defn{well-batched}
	if
	$\sizeof{\add{\seq[{2k}]}}=\sizeof{\add{\seq[{2k+1}]}}$ 
	for every 
	$2k\in \intset{0}{\sizeof\jpath-1}$.
	Each well-batched view $\view$ induces an equivalence relation $\strateq{\gG}{\view}$ over $\vertices[\gG]$ generated by: 
	%
	\begin{equation}
		u\strateqone{\gG}{\view}w
		\quad \mbox{ iff }
		\quad
		\mbox{$u=\fviewof{\view}^h_{2k}$ and $w=\fviewof{\view}^h_{2k+1}$ for a $2k < n-1$ and a $h\leq{\addsize{\view}}$}
	\end{equation}
	A \WIS $\strat$ is \defn{linked} if 
	it contains only well-batched views 
	and
	if for every $\view\in\strat$ the $\strateq{\gG}{\view}$-classes 
	are of the shape $\set{\isodd v_1, \dots, \isodd v_n, \iseven w}$.

	A \defn{$\CK$-winning innocent strategy} (or \CKWIS for short) is a linked
	\WIS $\strat$.
	\footnote{
		We here provide a simpler definition of \CKWISs w.r.t. the one in \cite{acc:cat:str:games}.
		In fact, we are able here to simplify this definition because we are considering the $\ldia$-free fragment of $\CK$.
	}
\end{definition}

\begin{example}
	Consider the arenas in \Cref{fig:exampleWIS}.
	The batched view of the (unique) maximal views in $\strat_1$ and $\strat_2$ are 
	$\begin{pmatrix}\emptyaddress&\isodd\lbox\\\iseven a&\isodd a\end{pmatrix}$
	and
	$\begin{pmatrix}
		\iseven \lbox_{10} & \isodd \lbox_{010}&\iseven \lbox_{000}&\iseven \lbox_{10} 
		\\ 
		\iseven b& \isodd b & \iseven a & \isodd a 
	\end{pmatrix}$
	respectively.
	The first is not well-batched because $\iseven a$ has height $0$ while $\isodd a$ has height $1$, 
	while 
	the second, even if well-batched, is not linked because
	the $\strateq{\gG}{\view}$-class 
	$\set{\iseven \lbox_{10} , \isodd \lbox_{010},\iseven \lbox_{000} }$
	contains two $\iseven\lbox$.
\end{example}

The definition of \CKWISs allows us to obtain a full-completeness result with respect to $\CK$ which, together with the good compositionality properties of \CKWISs shown in \cite{acc:cat:str:games,catta:phd}, provides a full-complete denotational semantics for the logic $\CK$.
That is, every given \CKWIS is the encoding of a derivation in $\CK$, and if a derivation $\dD$ reduces via cut-elimination to a derivation $\dD'$, then they are encoded by the same \CKWIS.

\begin{theorem}[\cite{acc:cat:str:games}]
	The set of \CKWISs is a full-complete denotational model for $\CK$.
\end{theorem}

\subsection{Full Completeness for Modal Lambda Terms in Normal Form}
\newcommand{\sameclass}[1][]{{\color{linkcolor}\approx}_{#1}}
\newcommand{\notsameclass}[1][]{{\not\approx}_{#1}}

We can prove the full completeness result using the type system $\FCK$ and relying on \Cref{thm:DNisFCK}.
For this purpose, we have to extend the definition of $\alpha$-equivalence from terms to type assignments in order to avoid technicality in our proofs, since in arenas we keep no track of variable names.
For example, consider the $\alpha$-equivalent terms $\lambda x.x$ and $\lambda y.y$ whose derivation should be considered non-equivalent due to the fact that $\alpha$-equivalence does not extends to type assignments, therefore the two occurrence of the axiom rule with conclusion $x:a\vdash x:a$ and $y:a\vdash y:a$ should be considered distinct.%
\footnote{
	Note that another possible way to deal with this problem is to label non-modal vertices of arenas by pairs of propositional atoms and variables instead of propositional variables only.
}

\def\aconeq{=_{\alpha}^{\Gamma;C}}
\begin{definition}
	Let $A_1,\ldots,A_n\vdash C$ be a sequent.
	We define 
	$\mltset(\Gamma\vdash C)$ as the set of terms $M$ such that the typing derivation $x_1:A_1,\ldots,x_n:A_n\vdash M:C$ is derivable, that is,
	$$
	\mltset(\Gamma\vdash C)=\Set{M\in\mltset\mid x_1:A_1,\ldots, x_n:A_n \vdash M:C \mbox{ is derivable for some }x_1,\ldots, x_n }
	\; .
	$$

	If $M,N \in \mltset(\Gamma\vdash C)$,
	we define $M \aconeq N $ as the smallest equivalence relation generated by the 
	rule
	$
	\vlinf{}{\text{$z_i$ is fresh}}{
		M \aconeq N 
	}{
		M \subst {z_1, \dots, z_n}{x_1, \dots, x_n} = N \subst{z_1, \dots, z_n}{y_1, \dots, y_n}
	}
	$.
\end{definition} 

From now on, we consider derivations up the $\alpha$-equivalence defined above, that is, 
we consider derivations up to renaming of the variables occurring in a typing context.

\begin{figure}[t]
	\adjustbox{max width=\textwidth}{$\begin{array}{c}
			\Den{
				\vlinf{\taxrule}{}{\Gamma,x:\isodd c\vdash x:\iseven c}{}
			}
			=
			\set{\emptyseq,\iseven c,\iseven c\isodd c}
			\qquad
			\Den{\vlderivation{\vlin{\mlimp}{}{
						\Gamma\vdash \lambda x_1^{A_1}\cdots x_n^{A_n}. M :  (A_1,\ldots A_n) \imp C
					}{\vlpr{\dD'}{}{\Gamma ,x_1 :  A_1, \ldots,x_n :  A_n \vdash M :  C}}}}
			=
			\Den{\dD'}
			\\\\
			\Den{
				\vlinf{\aximp}{}{
					\Gamma, y :  
					(A_1,\ldots,  A_n)\imp c
					\vdash yN_1\cdots N_n  :  c 
				}{
					\left\{
					\vlderivation{\vlpr{\dD_i}{}{\Gamma,
							(A_1,\ldots,  A_n)\imp c
							\vdash N_i :  A_i}}
					\right\}_{i\in\intset1n}
				}
			}
			=
			\Set{\emptyseq,\iseven c, \iseven c\isodd c}\cup 
			\Set{\iseven c\isodd c\jpath \mid \emptyseq\neq\jpath \in \Den{\dD_i}\mbox{ for a } i\in\intset1n }
			\\\\
			\Den{
				\vlderivation{
					\vlin{\exrule}{}{\Gamma \vdash M:C}{
						\vlpr{\dD'}{}{\sigma(\Gamma)\vdash M:C}
					}
				}
			}
			=
			\Set{f_\sigma(\view)\mid \view \in \Den{\dD'}, f_\sigma \mbox{ 	
					isomorphism between $\arof{\Gamma \vdash M:C}$ and $\arof{\sigma(\Gamma)\vdash M:C}$}}
			\\
			\mbox{where $f_\sigma(\view)$ is the view obtained by applying $f_\sigma$ to each move in $\view$ (and updating its pointer accordingly)}
			\\\\
			\Den{
				\vlinf{\kbrule}{}{
					\Delta, y_1 :   \lbox A_1,\ldots, y_n  :  A_n \vdash \lambdabox{M}{x_1,\ldots, x_n}{y_1,\ldots,y_n}  :  \lbox  C
				}{
					\vlderivation{\vlpr{\dD'}{}{x_1:A_1,\ldots,x_n:A_n \vdash M:C}}
				}
			}
			=
			\Den{\dD'}
			\\\\
			\Den{
				\vliinf{\kimp}{}{
					\Gamma,
					\underbrace{f_1 :  (A_{1,1},\ldots, A_{1,k_1})\imp \lbox B_1
						,\ldots, 
						f_n  :  (A_{n,1},\ldots, A_{n,k_n}) \imp \lbox B_n}_\Delta
					\vdash 
					\lambdabox{M}{N_1,\ldots, N_n,\myvec{z}}{y_1,\ldots y_n,\myvec{w}}  :  \lbox C 
				}{
					\left\{
					\vlderivation{\vlpr{\dD_{i,j}}{}{\Gamma,\Delta \vdash T_{i,j}  :   A_{i,j}}}
					\right\}_{
						i\in\intset1n,
						j\in\intset1{k_i}
					}
				}{	
					\vlderivation{
						\vlpr{\dD_0}{}{
							\Gamma, \Delta, x_1  :  \lbox  B_1, \ldots,  x_n  :  \lbox B_n 
							\vdash 
							\lambdabox{M}{x_1,\ldots,x_n, \myvec{z}}{y_1,\ldots , y_n,\myvec{w}}  :  \lbox C
						}
					}
				}
			}
			\\
			=
			\\
			\Den{\dD_0}
			\cup
			\left(
			\bigcup_{i\in\intset1n}
			\Set{\iseven c\isodd b_i \jpath \mid \emptyseq\neq\jpath\in\Den{\dD_{i,j}} \mbox{ for a } j\in\intset1{k_i}}
			\right)
			\\
			\mbox{where
				$\iseven c$ 
				(resp.~$\isodd b_i$)
				is the unique non-$\lbox$ vertex in $\irof[{\arof{\lbox C}}]$ (resp.~in $\arof{\lbox B_i}$).
			} 
		\end{array}$}
	\caption{
		Rules to construct a \CKWIS from a type derivation in $\FCK$. 
		For reasons of readability, we assume there is an implicit map identifying the moves in the arenas of the \as in the premises with the moves in the arena of the \as in the conclusion.
		Note that $\iseven c$ and $\isodd c$ are occurrences of the same atom $c$, but we have decorate them to improve readability.
	}
	\label{fig:deseq}
\end{figure}

\begin{restatable}{lemma}{lemWIStoFCK}\label{lem:WIStoFCK}
	Let 
	$\strat$ be a non-trivial \CKWIS over the arena $\arof{\Gamma\vdash C}$.
	Then there is a canonically defined $T_\strat\in\normTerms\cap \mltset(\Gamma \vdash C)$
	admitting a unique typing derivation in $\FCK$.
\end{restatable}
\begin{proof}
	We define a term $T_\strat$ and a derivation $\dD_\strat(T)$ in $\FCK$ of the \as $\Gamma \vdash T_\strat: C$
	by induction on the lexicographic order over the pairs $(\sizeof\strat, \sizeof C)$:
	\begin{enumerate}
		
		\item\label{lem:seq:i}
		if $C=c$ is atomic, then $\strat$ must contain the \CKWIS $\set{\emptyseq, \iseven c, \iseven c\isodd c}$ because $\iseven c$ is the unique \root of $\arof{\Gamma \vdash C}$ and $\isodd c$ is the unique (by $\oddsym$-totality) \oddmove justified by the unique  previous \evenmove in $\strat$.
		Note that by the well-batched condition we must have $\add{\iseven c}=\add{\isodd c}=\emptyaddress$.
		Then
		
		\begin{enumerate}
			
			\item 
			either 
			$\strat=\set{\emptyseq,c,cc}$,
			then
			$\Gamma= \Delta, c$
			for a sequent $\Delta$ such that no move in $\Delta$ occurs in $\strat$ (because of $\evensym$-completeness).
			In this case 
			$T=x$
			and
			$$\dD_\strat(T)=\vlinf{\AXrule}{}{\Delta, x:c\vdash x:c}{}$$
			
			\item\label{lem:seq:i.b}
			or,
			since $\strat$ is prefix-closed and well-batched,			
			$\strat= \left(\set{\emptyseq,\iseven c,\iseven c\isodd c} \bigcup_{i=0}^{n} cc\strat_i\right)$
			for some \CKWISs $\strat_i$ 
			for a sequent $\Gamma \vdash A_i$ 
			for each $i\in\intset1n$.
			Then
			$\Gamma=\Delta, (A_0, \ldots , A_n) \imp c$.
			In this case $T=yN_1\cdots N_n$
			and
			
			\
			
			\hskip-2em\adjustbox{max width=.9\textwidth}{$
				\dD_\strat(T)
				=
				\vlderivation{
					\vliiin{\aximp}{}{
						\Delta,y:(A_1,\ldots,  A_n)\imp c\vdash c
					}{
						\vlpr{\dD_{\strat_1}(N_1)}{}{\Delta, y:(A_1,\ldots,  A_n)\imp c\vdash N_1:A_1}
					}{
						\vlhy{\cdots}
					}{
						\vlpr{\dD_{\strat_n} (N_n)}{}{\Delta, y:(A_1,\ldots,  A_n)\imp c\vdash N_n:A_n} 
					}
				}
				$}
		\end{enumerate}
		
		\item 
		if
		$C=(A_1, \ldots , A_n) \imp B$,
		then 
		$T=\lambda x_1\cdots\lambda x_n.T'$
		and
		$$
		\dD_\strat(T)
		=
		\vlderivation{
			\vlin{\mlimp}{}{
				\Gamma\vdash T : (A_1, \ldots , A_n) \imp B
			}{
				\vlpr{\dD_{\strat}(T')}{}{
					\Gamma ,x_1:A_1, \ldots, x_n:A_n \vdash T' : B}
			}
		} 
		$$
		
		\item 
		$C=\iseven\lbox A$ is a $\lbox$-formula, 
		then,
		if a 
		\root $v$ of $\arof \Gamma$ occurs as a move in $\strat$, then it mus be justified by a \root of $\arof{\lbox A}$. 
		Therefore, 
		by the well-batched condition,  
		$v$ must be in the scope of a $\lbox$.
		We have two cases:
		
		\begin{enumerate}
			\item 
			either
			$\Gamma =\Sigma, \lbox \Delta$
			and no move in $\Sigma$ occurs in $\strat$.
			In this case we have that $T=\lambdabox{M}{x_1,\ldots, x_n}{y_1,\ldots,y_n}$
			and
			$$
			\dD_\strat(T)
			=
			\vlderivation{\vlin{\kbrule}{}{
					\Delta, y_1 :   \lbox A_1,\ldots, y_n  :  A_n \vdash \lambdabox{M}{x_1,\ldots, x_n}{y_1,\ldots,y_n}  :  \lbox  C
				}{
					\vlpr{\dD_{\strat}(M)}{}{x_1 : A_1,\ldots,x_n : A_n \vdash M : C}
				}
			}
			$$

			\item
			or
			$\Gamma=
			\Delta,
			(A_{1,1},\ldots A_{1,k_1})\imp \lbox B_1
			,\ldots, 
			(A_{n,1},\ldots A_{n,k_n})\imp \lbox B_n
			$
			for 
			some $k_1,\ldots, k_n$ such that $k_1+\cdots +k_n>0$
			and 
			a sequent $\Delta$ such that 
			if $\isodd \lbox D$ or $(A_1, \ldots, A_n)\imp \isodd\lbox D$ is in $\Delta$, then 
			$\isodd\lbox \notsameclass[\strat] \iseven\lbox$.
			In this case,
			by similar reasoning of (\ref{lem:seq:i}.\ref{lem:seq:i.b}),
			there are
			for some \CKWISs $\strat_i$ for the sequent $\Gamma \vdash A_i$ 
			for each $i\in\intset1n$
			and 
			a \CKWIS $\strat_0$ for the sequent $\Gamma, \Delta, \lbox B_1, \ldots, \lbox B_n\vdash \lbox C$
			such that
			$\strat= \left(\strat_0\cup \left(\bigcup_{i=0}^{n} \sigma_i\strat_i\right)\right)$.
			Therefore 
			$T=\lambdabox{M}{x_1,\ldots,x_n, \myvec{z}}{y_1,\ldots , y_n,\myvec{w}}$
			and
			$\dD_\strat(T)$ is the following derivation
			
			\
			
			\hspace{-3em}
			\scalebox{1.02}{$
				\vlderivation{
					\vliin{\kimp}{}{
						\Gamma, 
						\Delta
						\vdash 
						\lambdabox{M}{N_1,\ldots, N_n,\myvec{z}}{y_1,\ldots y_n,\myvec{w}}  :  \lbox C
					}{
						\vlhy{\left\{
							\vlderivation{
								\vlpr{\dD_{\strat_{i,j}}(T_{i,j})}{}{
									\Gamma,\Delta \vdash T_{i,j}:A_{i,j}
								}
							}
							\right\}_{
								\scriptsize
								\begin{array}{l}
									i\in\intset1n,
									\\
									j\in\intset1{k_i}
								\end{array}
							}
						}
					}{	
						\vlhy{
							\vlderivation{
								\vlpr{\dD_{\strat_0}\left(\lambdabox{M}{x_1,\ldots,x_n, \myvec{z}}{y_1,\ldots , y_n,\myvec{w}}\right)}{}{
									\Gamma, \Delta, 
									\Delta'
									\vdash \lambdabox{M}{x_1,\ldots,x_n, \myvec{z}}{y_1,\ldots , y_n,\myvec{w}}  :  \lbox C}
							}
						}
					}
				}
				$}
			\
			
			where
			$\Delta=f_1 :  (A_{1,1},\ldots, A_{1,k_1})\imp \lbox B_1
			,\ldots, 
			f_n  :  (A_{n,1},\ldots, A_{n,k_n}) \imp \lbox B_n$
			and
			$\Delta'=x_1  :  \lbox  B_1, \ldots,  x_n  :  \lbox B_n $.
		\end{enumerate}
	\end{enumerate}
\end{proof}
\begin{theorem}
	There is a one-to-one correspondence between terms in  $\normTerms \cap \mltset(\Gamma\vdash C)$
	and
	\CKWIS for $\Gamma \vdash C$.
\end{theorem}
\begin{proof}
	\Cref{lem:WIStoFCK} ensures that one \CKWIS $\strat$ for $\Gamma \vdash C$,
	we can define a (unique) typing derivation $\dD_\strat$ in $\FCK$ 
	of a term $T_\strat\in \normTerms \cap \mltset(\Gamma\vdash C)$.
	
	Conversely, given a
	\as
	$\Gamma \vdash T: C$
	for a $T\in\normTerms$, 
	then, we can uniquely define is a derivation $\dD_T$ in $\FCK$.
	Thus,  by \Cref{thm:DNisFCK},  the \as $\Gamma \vdash T: C$ is derivable.
	We define $\strat_T$ as the \CKWIS defined by induction on the number of rules in $\dD_T$ using the rules in \Cref{fig:deseq}.
	
	We conclude since we have that $\strat_{T_\strat}=\strat$ and $T_{\strat_T}=T$ by definition.
\end{proof}

\section{Conclusion}\label{sec:out}

In this paper we introduced a new \mlcalc for the $\ldia$-free fragment of the constructive modal logic $\CK$ (without conjunction or disjunction).
This lambda calculus builds on the work in \cite{kakutanilambda},
by adding a restricted $\eta$-reduction as well as two new reduction rules dealing with the explicit substitution constructor used to model the modality $\lbox$.  
We proved normalization  and confluence for this calculus 
and 
we provide a one-to-one correspondence between the set of terms in normal form and the set of winning strategies for the logic $\CK$ introduced in \cite{acc:cat:str:games}. 

We foresee the possibility of extending the result presented in this paper to the entire disjunction-free fragment of $\CK$, for which winning strategies are already defined in \cite{acc:cat:str:games} are a fully complete denotational semantics.  For this purpose, we should consider additional term constructors for terms whose type is a conjunction, as well as a new $\letcon$-like operator to model terms whose type is the modality $\ldia$-formula similar to the one proposed in \cite{bel:deP:rit:extended}.
For this reason, 
in future works we plan to reformulate our lambda-calculus in the light of the novel line of research on calculi with explicit substitutions \cite{kes:explicit,kes:LMCS2009,acc:bon:kes:lom,acc:exponentials}.
This approach would allow us to simplify some of the technicalities and achieve a more elegant operational semantics. 
Another interesting prospective is to extend our approach to operational semantics to the Fitch-style \mlcalc studied in \cite{nachiappa}.

At the same time, we plan to make explicit that our game semantics provides a concrete model for the 
\emph{cartesian closed categories} provided with a \emph{strong monoidal endofunctor} \cite{bel:deP:rit:extended,Kavvos20}.
Indeed, 
categorical semantics of the calculus in \cite{bel:deP:rit:extended} is modeled by means of 
\emph{cartesian closed categories} equipped with a \emph{strong monoidal endofunctor} 
taking into account the proof-theoretical behavior of the $\lbox$-modality. 
We further conjecture that the syntactic category obtained via the quotient of modal terms modulo the relations we introduce in this paper
is indeed a \emph{free cartesian closed category} on a set of atoms with a \emph{strong monoidal endofunctor}.

%
\bibliographystyle{splncs04}
\bibliography{ref}

\end{document}

%% file: macros.tex
\def\set#1{\{#1\}}
\def\Set#1{\left\{#1\right\}}
\def\Setof#1#2{\Set{\begin{array}{c|l}#1&#2\end{array}}}
\def\intset#1#2{\set{#1,\ldots, #2}}

\def\sizeof#1{|#1|}
\def\sizeinof#1#2{|#1|_{#2}}

\def\arof#1{\left[\!\left[#1\right]\!\right]}

\def\imp{\to}
\def\cons#1{\left[#1\right]}
\def\chole{\circ}
\def\conso{\cons\chole}
\newcommand{\CwH}[1][]{\mathsf{CwH_{{#1}}}}

\def\alphaeq{=_\alpha}

\def\FV#1{\texttt{FV}(#1)}
\def\lone{1}

\def\normTerms{\widehat{ \Lambda}}
\def\nfset{\nf{\Lambda}{\beta \eta \kappa}}
\def\as{type assignment\xspace}

\def\are{\mathsf A}

\def\Den#1{\left\{\!\left\{#1\right\}\!\right\}}
\def\evensym{\circ}
\def\oddsym{\bullet}
\def\iseven#1{#1^{\evensym}}
\def\isodd#1{#1^{\oddsym}}

\def\strat{\mathcal S}

\def\add#1{\mathsf{add}({#1})}
\newcommand{\addi}[2]{\mathsf{add}^{#1}({#2})}
\def\addsize#1{{\color{modcolor} h}_{#1}}
\def\emptyaddress{\epsilon}

\def\depsym{{\color{arenacolor} d}}
\def\dep#1{\depsym(#1)}

\def\evenmove{$\evensym$-move\xspace}
\def\oddmove{$\oddsym$-move\xspace}

\def\jpath{\mathsf p}

\def\view{\mathsf{p}}
\def\viewi#1{\view_{#1}}
\def\fviewof#1{\mathcal F(#1)}
\def\emptyseq{\epsilon}

\newcommand{\axeq}[1][]{\mkern1mu\mathord{\stackrel{#1}{{\color{blue}\sim}}}\mkern1mu}

\newcommand{\strateq}[2]{\axeq[{#1}_{#2}]}

\newcommand{\strateqone}[2]{\axeq[{#1}_{#2}]_1}

\def\emptyseq{\epsilon}

\def\dD{\mathcal{D}}

\def\CK{\mathsf{CK}}
\def\SCK{\mathsf{SCK}}

\def\FCK{\mathsf{CK^F}}
\def\NDCK{\mathsf{ND}_{\CK}}

\def\dnbox{\mathsf{\lbox\text{-}subst}}
\def\taxrule{\mathsf{ax}} 		
\def\AXrule{\mathsf{ax}}
\def\exrule{\mathsf{ex}}
\def\Crule{\mathsf{C}}
\def\Wrule{\mathsf{W}}
\def\cutr{\mathsf{cut}}

\def\rimprule{\imp^{\mathsf{R}}}
\def\limprule{\imp^{\mathsf{L}}}

\def\kbrule{\mathsf{K}^\lbox}

\def\Idrule{\mathsf {Id}}
\def\impintro{ \mathsf{Abs}}
\def\impelim{ \mathsf{App} }

\def\atoms{\mathcal A}
\def\variables{\mathcal V}

\def\typed#1#2{#1^{#2}}

\def\tuple#1{\langle #1\rangle }

\def\pair#1#2{\tuple{#1, #2}}

\def\tlambda#1#2{\lambda {#1}^{#2}}

\newcommand{\lambdalet}[3]{\mathsf{Let}\;#1\;\mathsf{be}\;#2\;\mathsf{in}\;#3}

\newcommand{\llambda}[3][]{\lambda \typed{#2}{#1}.{#3}}

\def\boxpar#1{\left[#1\right]_\blacksquare}

\def\substpar#1{\left\{#1\right\}}

\def\subst#1#2{\substpar{#1/#2}}
\def\lambdabox#1#2#3{ #1 \boxpar{#2/#3}}

\def\etao#1{ \|#1\|_\eta^1 }
\def\etat#1{ \|#1\|_\eta^2 }
\def\etam#1{ \|#1\|_\eta}

\def\sizeo#1{\etao{#1}}
\def\sizet#1{\etat{#1}}
\def\kappam#1{  \|#1\|_\kappa  }
\def\nf#1#2{  \mathsf{nf}_{#2} (#1)   }

\def\mlt{modal $\lambda$-term\xspace}
\def\mlts{modal $\lambda$-terms\xspace}
\def\mltset{\Lambda}
\def\smltset{\Lambda^\blacksquare}
\def\amltset{\Lambda^\lambda}

\def\subterm#1{\mathsf{Sub}(#1)}

\newcommand{\bred}[1][]{\rightsquigarrow_{\beta_{#1}}}
\newcommand{\ered}[1][]{\rightsquigarrow_{\eta_{#1}}}
\newcommand{\kred}[1][]{\rightsquigarrow_{\kappa_{#1}}}
\def\eredo{\ered[1]}
\def\eredt{\ered[2]}

\def\bered{\rightsquigarrow_{\beta\eta}}
\def\bkred{\rightsquigarrow_{\beta\kappa}}
\def\bekred{\rightsquigarrow_{\beta\eta\kappa}}

\def\reflc#1{\rightsquigarrow^=_{#1}}

\def\transc#1{\rightsquigarrow^+_{#1}}
\def\rtransc#1{   \rightsquigarrow^\ast_{#1}   }

\def\cC{ \mathbf  C} 
\def\cE{ \mathbf{E}}
 \def\cD{ \mathbf{D}}

\def\WIS{$\mathsf{WIS}$\xspace}
\def\WISs{$\mathsf{WIS}$s\xspace}

\def\CKWIS{$\CK$-\WIS}
\def\CKWISs{$\CK$-\WISs}

\def\quand{\quad\mbox{and}\quad}
\def\rclr#1{{\color{red}#1}}

\usepackage{marginnote}

\def\gG{\mathcal G}

\def\gH{\mathcal H}

\newcommand{\vertices}[1][]{V_{#1}}

\newcommand{\dedge}[1][]{\mkern1mu\mathord{\stackrel{#1}{\curvearrowright}}\mkern1mu}
\newcommand{\iedge}[1][]{\mkern1mu\mathord{\stackrel{#1}{{\color{arenacolor}\to}}}\mkern1mu}
\newcommand{\medge}[1][]{\mkern1mu\mathord{\stackrel{#1}{{\color{modcolor}\rightsquigarrow}}}\mkern1mu}

\makeatletter
\newcommand{\oset}[3][.8ex]{%
	\mathrel{\mathop{#3}\limits^{
			\vbox to#1{\kern-.5\ex@
				\hbox{$\scriptstyle#2$}\vss}}}
}
\makeatother

\newcommand{\irof}[1][]{\smash{\oset{\iedge}{R}_{\!#1}}}

\def\ddag{\textsf{2-dag}\xspace}
\def\ddags{\textsf{2-dag}'s\xspace}

\def\labelset{\mathcal L}
\def\labsym{\ell}
\def\lab#1{\labsym(#1)}

\newcommand{\ssbox}{\lbox}
\newcommand{\singlevertex}[1][]{{#1}}

\newcommand{\RtoR}[2]{\left(\irof[#1]\dedge\irof[#2]\right)}

\def\gimp{{\color{arenacolor}-\hspace{-2pt}\lcoseq}}
\def\gmod{{\color{modcolor}\sim\hspace{-4.2pt}\lcoseq}}